\documentclass[a4paper,11pt]{article}
\usepackage[utf8x]{inputenc}
\usepackage{fancyhdr}
\usepackage{graphicx}
\usepackage{amsmath}
\usepackage{amssymb}
\usepackage{hyperref}
\usepackage[all]{hypcap}
\usepackage{pdfpages}
\usepackage{mathrsfs}
\usepackage{caption}
\usepackage{subcaption}
\usepackage[margin=0.95in]{geometry}
\newtheorem{theorem}{Theorem}[section]
\newtheorem{lemma}[theorem]{Lemma}
\newtheorem{proposition}[theorem]{Proposition}

\newenvironment{proof}[1][Proof]{\begin{trivlist}
\item[\hskip \labelsep {\bfseries #1}]}{\end{trivlist}}
\newenvironment{definition}[1][Definition]{\begin{trivlist}
\item[\hskip \labelsep {\bfseries #1}]}{\end{trivlist}}

\newcommand{\qed}{\nobreak \ifvmode \relax \else
      \ifdim\lastskip<1.5em \hskip-\lastskip
      \hskip1.5em plus0em minus0.5em \fi \nobreak
      \vrule height0.75em width0.5em depth0.25em\fi}
\title{Stochastic Recursive Inclusions with Non-Additive Iterate-Dependent Markov Noise}
\author{Vinayaka G. Yaji and Shalabh Bhatnagar}
\begin{document}
\maketitle
\begin{abstract}
 In this paper we study the asymptotic behavior of stochastic approximation schemes with set-valued drift function and non-additive 
 iterate-dependent Markov noise. We show that a linearly interpolated trajectory of such a recursion is an asymptotic pseudotrajectory for the flow of a 
 limiting differential inclusion obtained by averaging the set-valued drift function of the recursion w.r.t. the stationary distributions of 
 the Markov noise. The limit set theorem in \cite{benaim1} is then used to characterize the limit sets of the recursion in terms of the dynamics of 
 the limiting differential inclusion. We then state two variants of the Markov noise assumption under which the analysis of the recursion is 
 similar to the one presented in this paper. Scenarios where our recursion naturally appears are presented as applications. These include 
 controlled stochastic approximation, subgradient descent, approximate drift problem and analysis of discontinuous dynamics all in the presence 
 of non-additive iterate-dependent Markov noise.
\end{abstract}

\section{Introduction}
Consider the standard stochastic approximation scheme given by, 
\begin{equation}\label{recstd}
 X_{n+1}=X_n+a(n)\left(h(X_n)+M_{n+1}\right)
\end{equation}
where $\left\{X_n\right\}_{n\geq0}$ is a sequence of $\mathbb{R}^d$-valued random variables, $h:\mathbb{R}^d\rightarrow\mathbb{R}^d$ is a Lipschitz 
continuous drift function and $\left\{M_{n}\right\}_{n\geq1}$ is a sequence of $\mathbb{R}^d$-valued random variables which denote the additive noise 
terms. In \cite{benaim2}, it was shown that under certain assumptions the asymptotic behavior of recursion \eqref{recstd} can be determined by 
the asymptotic behavior of the \it{o.d.e.}\rm 
\begin{equation*}
 \frac{dx}{dt}=h(x).
\end{equation*}
This method is known as the ODE method and the central idea of this method is to show that the linearly interpolated trajectory of recursion 
\eqref{recstd} \it{``tracks'' }\rm the flow of the o.d.e. This idea was later generalized in \cite{benaim3} to analyze the asymptotic behavior 
of stochastic processes with continuous sample paths for which a tracking argument as in the ODE method could be established. It was here that 
the notion of \it{asymptotic pseudotrajectory} \rm was introduced. For a precise definition we refer the reader to chapter 3 of \cite{benaim3}.
In many applications arising in machine learning and optimization, the drift function $h$ is set-valued or is single-valued and does not 
satisfy the assumption of Lipschitz continuity (or even continuity). In such cases the recursion studied takes the form,
\begin{equation}\label{recsvm}
 X_{n+1}-X_{n}-a(n)M_{n+1}\in a(n)H(X_n),
\end{equation}
where $H:\mathbb{R}^d\rightarrow\left\{\text{subsets of }\mathbb{R}^d\right\}$ is a set-valued map satisfying certain assumptions. Such 
recursions were first studied in \cite{benaim1} by showing that the linearly interpolated trajectory of recursion \eqref{recsvm} \it{``tracks'' }\rm 
the flow of the differential inclusion given by,
\begin{equation*}
 \frac{dx}{dt}\in H(x).
\end{equation*}
The notion of asymptotic pseudotrajectory was also extended to the set-valued case and a limit set theorem established.

In many applications such as methods relying on Monte Carlo simulation where sampling exactly from a distribution (or family of distributions) 
is not possible and instead Markov chain Monte Carlo methods are used, the recursion consists of an additional non-additive iterate-dependent 
Markov noise component. The recursion takes the general form,
\begin{equation}\label{recmark1}
 X_{n+1}-X_n=a(n)\left(h(X_n,S_n)+M_{n+1}\right),
\end{equation}
where $\{S_n\}_{n\geq0}$ denotes the Markov noise component. The other quantities in the above recursion have the same interpretation as those 
in \eqref{recstd}. The analysis of such a recursion was first performed in \cite{benven}. In \cite{benven}, the analysis is 
carried out under the assumptions that the drift function $h$ is continuous, the Markov chain defined by the transition kernel associated with 
the Markov noise terms in the recursion admits a unique stationary distribution as well as a solution to the Poisson equation among others. The 
case with discontinuities in the drift function was considered in \cite{tadic} but the analysis required the function $h(\cdot,s)$ to be Holder continuous. 
Recently in \cite{matti}, such a recursion has been analyzed when the drift function is just measurable. In all the above references, however, the Markov 
noise terms are assumed to have a unique stationary distribution and the assumptions imply an existence of a solution to the Poisson equation. 
Another set of assumptions on the Markov noise terms under which the recursion \eqref{recmark1} is studied can be found in \cite{borkarmark}. 
In \cite{borkarmark}, the transition kernel defining the Markov noise terms is required to be continuous in both the iterate and the state variables. 
The advantage in this case is that the set of stationary distributions need not be unique and one does not need the Markov chain associated with 
the transition kernel for every value of the iterate, to be aperiodic and irreducible. Further the analysis in \cite{borkarmark} enables one to study 
the recursion \eqref{recmark1} when the noise terms are not Markov by themselves, but their lack of Markov property comes through the dependence on 
a control sequence.

In this paper we extend the above to the case where the drift function $h$ in recursion \eqref{recmark1}, is a set-valued map. The recursion now 
takes the form,
\begin{equation*}
 X_{n+1}-X_n-a(n)M_{n+1}\in a(n)H(X_n,S_n),
\end{equation*}
where $H:\mathbb{R}^d\times\mathcal{S}\rightarrow\left\{\text{subsets of }\mathbb{R}^d\right\}$ is a set-valued map and $\mathcal{S}$ denotes the 
state space of the Markov noise terms which we assume to be a compact metric space. The assumption on the Markov noise terms is similar to 
\cite{borkarmark}. We shall show that the linearly interpolated trajectory of the above recursion is an asymptotic pseudotrajectory for the flow 
of a limiting differential inclusion obtained by averaging the set-valued map w.r.t. the stationary distributions of the Markov noise terms. 
The main idea is to approximate the set-valued drift function with continuous set-valued maps which admit a single-valued parametrization, there 
by enabling us to write the recursion as a standard stochastic approximation scheme with single-valued maps for which an asymptotic 
pseudotrajectory argument is easy to establish. Later we invoke the limit set theorem in \cite{benaim1} to characterize the limit set of 
the above recursion.\newline
\it{Organization of the paper:}\rm

In section \ref{rec_ass}, we formally define the recursion and state the assumptions imposed, which is then followed by a brief discussion on each of 
the assumptions. 
In section \ref{backgrd}, we review certain results from set-valued analysis which will be used later to analyze the recursion and define the 
limiting differential inclusion.
In section \ref{ldi}, the limiting differential inclusion is defined and properties of the same are established.
In section \ref{sopmvf}, we define the space of probability measure valued functions and define an appropriate topology on this space. A metrization 
lemma for the above mentioned space is stated and a continuous function on such a space needed later in analyzing the recursion is defined.
In section \ref{recanal}, we first state some preliminaries and later state our main result followed by the limit set theorem.
In section \ref{siacmnc}, we state two variants of the Markov noise assumption under which the analysis of the stochastic approximation scheme 
can be carried out along similar lines.
In section \ref{app}, we state four applications where the recursion analyzed in this paper naturally appears. We finally conclude in section 
\ref{cadffw} by providing a few interesting directions for future work.

  \section{Recursion and assumptions}
  \label{rec_ass}
  Let $(\Omega,\mathscr{F},\mathbb{P})$ be a probability space and $\left\{X_n\right\}_{n\geq0}$ be a sequence of $\mathbb{R}^d$-valued 
  random variables on $\Omega$ satisfying
  \begin{equation}\label{rec}
  X_{n+1}-X_{n}-a(n)M_{n+1}\in a(n)H(X_n,S_n)
  \end{equation}
  where,
  \begin{itemize}
  \item [(A1)]$H:\mathbb{R}^d\times\mathcal{S}\rightarrow\left\{\text{subsets of }\mathbb{R}^d\right\}$ is a set valued map on 
  $\mathbb{R}^d\times\mathcal{S}$ where $(\mathcal{S},d_{\mathcal{S}})$ is a compact metric space and the map $H(\cdot)$ satisfies
      \begin{itemize}
      \item [(i)] for every $(x,s)\in\mathbb{R}^d\times\mathcal{S}$, $H(x,s)$ is a non-empty, compact, convex subset of $\mathbb{R}^d$,
      \item [(ii)] there exists $K>0$ such that, for every 
      $(x,s)\in\mathbb{R}^d\times\mathcal{S},\\ \sup_{z\in H(x,s)}\parallel z\parallel\leq K(1+\parallel x\parallel)$,
      \item [(iii)] for every $\mathbb{R}^d\times\mathcal{S}$ valued sequence, $\left\{\left(x_n,s_n\right)\right\}_{n\geq1}$ converging to $(x,s)$, for every 
      sequence $\left\{z_n\in H(x_n,s_n)\right\}_{n\geq1}$ converging to $z$, we have, $z\in H(x,s)$.  
      \end{itemize}
  \end{itemize}
      
  Throughout this paper $\mathcal{P}(\cdots)$, denotes the set of probability measures on a compact metric space \lq$\cdots$' with the 
  Prohorov topology (see chapter 2 in \cite{borkarap}).
  \begin{itemize}
  \item [(A2)]$\left\{S_n\right\}_{n\geq0}$ is a sequence of $\mathcal{S}$ valued random variables on $\Omega$ such that, 
  for every $n\geq0$, for every $A\in \mathscr{B}(\mathcal{S})$, $\mathbb{P}(S_{n+1}\in A|S_m,X_m,m\leq n)=
  \mathbb{P}(S_{n+1}\in A|S_n,X_n)=\Pi(X_n,S_n)(A)$ $a.s.$ where $\Pi:\mathbb{R}^d\times\mathcal{S}\rightarrow\mathcal{P}(\mathcal{S})$, 
  is continuous.
  \item [(A3)] $\left\{a(n)\right\}_{n\geq0}$ is a sequence of real numbers satisfying
  \begin{itemize}
  \item[(i)]$a(0)\leq1$ and for every $n\geq0$, $a(n)\geq a(n+1)$,
  \item[(ii)] $\sum_{n\geq0}a(n)=\infty$ and $\sum_{n\geq0}a(n)^2<\infty$.
  \end{itemize}
  \item [(A4)] $\left\{M_n\right\}_{n\geq1}$ is a sequence of $\mathbb{R}^d$ valued random variables such that, $a.e. (\omega)$, for every $T>0$ 
  $\\\lim_{n\to\infty}\sup_{n\leq k\leq\tau(n,T)}\parallel\sum_{m=n}^{k}a(m)M_{m+1}\parallel=0$ where, 
  $\tau(n,T):=\min\left\{m>n:\sum_{k=n}^{m-1}a(k)\geq T\right\}$.
  \item[(A5)] $\mathbb{P}(\sup_{n\geq0}\parallel X_n \parallel<\infty)=1$.
  \end{itemize}
  
Assumption $(A1)$ is an extension of the assumption imposed on the set valued map for the case without Markov noise in \cite{benaim1} to the 
case with Markov noise. The only difference is in $(A1)(ii)$, since in our case we expect the constant $K$ to be independent of Markov noise 
terms. This strengthening of the assumption allows us to obtain integrable measurable selections and also to prove certain properties of the 
limiting differential inclusion in future sections.

Assumption $(A2)$ ensures that for every $x\in \mathbb{R}^d$ the Markov chain defined by the transition kernel $\Pi(x,\cdot)(\cdot)$ possesses 
weak Feller property (see chapter 6.1.1 of \cite{meyntweed}). In addition to the above, since the state space of the Markov chain is a compact 
metric space, by Theorem 12.1.2(ii) in \cite{meyntweed}, for every $x\in\mathbb{R}^d$, the Markov chain defined by the transition kernel 
$\Pi(x,\cdot)(\cdot)$ admits at least one stationary distribution ($\mu\in\mathcal{P}(\mathcal{S})$ is stationary for the Markov chain defined 
by the transition kernel $\Pi(x,\cdot)(\cdot)$ if, for every $A\in\mathscr{B}(\mathcal{S})$, $\mu(A)=\int_{\mathcal{S}}\Pi(x,s)(A)\mu(ds)$).
For every $x\in\mathbb{R}^d$, let $D(x)\subseteq\mathcal{P}(\mathcal{S})$ denote the set of stationary distributions of the Markov chain 
defined by the transition kernel $\Pi(x,\cdot)(\cdot)$. It can be easily shown that,  
\begin{itemize}
 \item [(i)] for every $x\in\mathbb{R}^d$, $D(x)$ is a convex and compact subset of $\mathcal{P}(\mathcal{S})$.
 \item [(ii)] the map $x\rightarrow D(x)$ has a closed graph, that is the set 
 \begin{equation*}
  \left\{(x,\mu)\in\mathbb{R}^d\times\mathcal{P}(\mathcal{S}): x\in\mathbb{R}^d,\ \mu\in D(x)\right\},
 \end{equation*}
is a closed subset of $\mathbb{R}^d\times\mathcal{P}(\mathcal{S})$(for a proof of the above two properties we refer the reader to page 
69 in \cite{borkartxt}). 
\end{itemize}
It is worth noting here that for every $x\in\mathbb{R}^d$, we do \it{not }\rm require that the associated Markov chain be aperiodic, 
irreducible or possess a unique stationary distribution.

Assumption $(A3)$ is the standard step size assumption. The assumptions of a non-increasing sequence and square summability are required in 
the analysis of the Markov noise component.

Assumption $(A4)$ is the same additive noise assumption as in Proposition 1.3, section 1.5 of \cite{benaim1}. The assumption ensures that the 
contribution of the additive noise is eventually negligible. For various sufficient conditions under which $(A4)$ is satisfied we refer the 
reader to \cite{benaim1}.

Assumption $(A5)$ ensures that the sample path of iterates of recursion \eqref{rec} remain bounded. This assumption is known as the stability 
assumption and sufficient conditions for ensuring the same for the single-valued case (see \cite{borkarmeyn}, \cite{andvih})  and for 
the set-valued case without Markov noise(see \cite{arunstab}) exist.
\section{Background}
\label{backgrd}
In this section we shall review some results from the theory of set valued maps and their integration. First we shall present a result on 
approximation of upper semi-continuous set valued maps with continuous set valued maps. We will then use this approximation to obtain 
single valued maps that represent these approximate continuous set valued maps thus enabling us to rewrite the stochastic recursive 
inclusion, \eqref{rec}, as a standard stochastic approximation scheme with single valued maps.

Secondly, in order to define the set valued map associated with the limiting differential inclusion (also known as mean field) and prove 
properties of the same we review definitions of measurability of set valued maps, its integral and some properties of the integral.

\subsection{Upper semicontinuous set valued maps and their approximation}
 First we shall recall the notion of upper semicontinuity, lower semicontinuity and continuity of set valued maps. These notions are 
 taken from chapter 1, section 1 of \cite{aubindi}.  
\begin{definition}
 A set valued map $H:\mathbb{R}^d\times\mathcal{S}\rightarrow\left\{\text{subsets of }\mathbb{R}^d\right\}$ is,
 \begin{itemize}
  \item \it{Upper semicontinuous }\rm(u.s.c.) if, for every $(x_0,s_0)\in\mathbb{R}^d\times\mathcal{S}$, for every $\epsilon>0$, there 
  exists $\delta>0$ (depending on $(x_0,s_0)$ and $\epsilon$) such that,
  \begin{equation*}
   \parallel x-x_0\parallel<\delta,\ d_{\mathcal{S}}(s,s_0)<\delta\implies H(x,s)\subseteq H(x_0,s_0)+\epsilon U,
  \end{equation*}
  where $U$ denotes the closed unit ball in $\mathbb{R}^d$.
  \item \it{Lower semicontinuous }\rm(l.s.c)  if, for every $(x_0,s_0)\in\mathbb{R}^d\times\mathcal{S}$, for every $z_0\in H(x_0,s_0)$, 
  for every sequence $\left\{\left(x_n,s_n\right)\right\}_{n\geq1}$ converging to $(x_0,s_0)$, there exists a sequence 
  $\left\{z_n\in H(x_n,s_n)\right\}$ converging to $z_0$.
  \item \it{Continuous }\rm if, it is both u.s.c. and l.s.c.
 \end{itemize}
\end{definition}
For set valued maps taking compact set values we have the above mentioned notion of u.s.c. to be equivalent to the standard notion of u.s.c.
(see pg. 45, \cite{aubindi}). In this paper we shall encounter set valued maps which are compact set valued and hence we have chosen to state 
the above as the definition of upper semicontinuity.

Since the set valued map $H$ satisfying assumption $(A1)$ has closed graph (i.e. assumption $(A1)(iii)$ holds), the following lemma follows 
from Corollary 1 in chapter 1, section 1 of \cite{aubindi}.

\begin{lemma}
 \emph{[u.s.c.]} A set valued map $H$ satisfying $(A1)$ is u.s.c.
\end{lemma}

The next lemma gives a sequence of decreasing continuous set valued maps which approximate the set valued map $H$ satisfying assumption 
$(A1)$. The statement of the lemma below can be found in page 39 of \cite{aubindi}.
\begin{lemma}
 \emph{[continuous embedding]}\label{ctem} Let $H$ be a set valued map satisfying $(A1)$. Then there exists a sequence of set valued 
 maps, $\left\{H^{(l)}\right\}_{l\geq1}$ where for every $\ l\geq1$, 
 \begin{itemize}
  \item [(i)] $H^{(l)}:\mathbb{R}^d\times\mathcal{S}\rightarrow\left\{\text{subsets of }\mathbb{R}^d\right\}$ is continuous and  for 
  every $(x,s)\in\mathbb{R}^d\times\mathcal{S}$, $H^{(l)}(x,s)$ is a convex and compact subset of $\mathbb{R}^d$,
  \item [(ii)] for every $(x,s)\in\mathbb{R}^d\times\mathcal{S}$, $H(x,s)\subseteq H^{(l+1)}(x,s)\subseteq H^{(l)}(x,s)$,
  \item [(iii)] there exists $K^{(l)}>0$, such that for every 
  $(x,s)\in\mathbb{R}^d\times\mathcal{S},\\ \sup_{z\in H^{(l)}(x,s)}\parallel z\parallel\leq K^{(l)}(1+\parallel x\parallel)$.
 \end{itemize}
 Furthermore, for every $(x,s)\in\mathbb{R}^d\times\mathcal{S}$, $H(x,s)=\cap_{\substack{l\geq1}}H^{(l)}(x,s)$.
\end{lemma}
A brief outline of the proof of the above lemma is provided in appendix \ref{prf_ctem}. 
The following are some useful observations from the proof of Lemma \ref{ctem}.
\begin{itemize}
 \item [(a)] $\tilde{K}:=\sup_{l\geq1}K^{(l)}$ is finite.
 \item [(b)] For every $(x,s)\in\mathbb{R}^d\times\mathcal{S}$, for every $\epsilon>0$, there exists $L$ (depending on $\epsilon$ and 
 $(x,s)$), such that, for every $l\geq L$, $H^{(l)}(x,s)\subseteq H(x,s)+\epsilon U$.
\end{itemize}

The next lemma gives us a continuous parametrization of the continuous set valued maps $H^{(l)}$ as in Lemma \ref{ctem}. This  
lemma follows from Theorem 2 in chapter 1, section 7 of \cite{aubindi}.
\begin{lemma}
 \emph{[parametrization]}\label{param} For every $l\geq1$, $H^{(l)}$ be a set valued map as in Lemma \ref{ctem}. Then for every 
 $l\geq1$, there exists a continuous function $h^{(l)}:\mathbb{R}^d\times\mathcal{S}\times U\rightarrow\mathbb{R}^d$ such that,
 \begin{itemize}
  \item [(i)] for every $(x,s)\in \mathbb{R}^d\times\mathcal{S}$, $H^{(l)}(x,s)=h^{(l)}(x,s,U)$ where $U$ denotes the closed unit ball 
  in $\mathbb{R}^d$ and $h^{(l)}(x,s,U):=\left\{h^{(l)}(x,s,u):u\in U\right\}$,
  \item[(ii)] for every $(x,s,u)\in\mathbb{R}^d\times\mathcal{S}\times U$, 
  $\parallel h^{(l)}(x,s,u)\parallel\leq K^{(l)}(1+\parallel x\parallel)$ where $K^{(l)}$ is as in Lemma \ref{ctem}\it{(iii)}\rm. 
 \end{itemize}
\end{lemma}

Combining Lemma \ref{ctem} and Lemma \ref{param} we obtain the approximation theorem stated below.
\begin{theorem}
 \emph{[approximation]} Let $H$ be a set valued map satisfying $(A1)$. Then there exists a sequence of continuous functions, 
 $\left\{h^{(l)}\right\}_{l\geq1}$, such that, for every $l\geq1$,
 \begin{itemize}
  \item [(i)] $h^{(l)}:\mathbb{R}^d\times\mathcal{S}\times U\rightarrow \mathbb{R}^d$ is continuous and for every 
  $(x,s)\in\mathbb{R}^d\times\mathcal{S}$, $h^{(l)}(x,s,U)$ is a convex and compact subset of $\mathbb{R}^d$,
  \item [(ii)] for every $(x,s)\in\mathbb{R}^d\times\mathcal{S}$, $H(x,s)\subseteq h^{(l+1)}(x,s,U)\subseteq h^{(l)}(x,s,U)$,
  \item [(iii)] there exists $K^{(l)}>0$, such that for every $(x,s,u)\in \mathbb{R}^d\times\mathcal{S}\times U$, 
  $\\\parallel h^{(l)}(x,s,u)\parallel\leq K^{(l)}(1+\parallel x\parallel)$.
 \end{itemize}
Furthermore, for every $(x,s)\in\mathbb{R}^d\times\mathcal{S}$, $H(x,s)=\cap_{l\geq1}h^{(l)}(x,s,U)$.
\end{theorem}

\subsection{Measurable set valued maps and integration}

Let $(\mathcal{Y},\mathscr{F}_{\mathcal{Y}})$ denote a measurable space and $F:\mathcal{Y}\rightarrow\left\{\text{subsets of }\mathbb{R}^d\right\}$
be a set valued map such that, for every $y\in\mathcal{Y}$, $F(y)$ is a non-empty closed subset of $\mathbb{R}^d$. Throughout this 
section $F$ refers to the set valued map as defined above.

\begin{definition}
\emph{[measurable set-valued map]} A set valued map $F$ is measurable if for every $C\subseteq \mathbb{R}^d$, closed,
\begin{equation*}
 F^{-1}(C):=\left\{y\in\mathcal{Y}:F(y)\cap C\neq\emptyset\right\}\in \mathscr{F}_{\mathcal{Y}}.
\end{equation*}
\end{definition}
We refer the reader to Theorem 1.2.3 in \cite{shoumei} for other notions of measurability and their relation to the definition above.

\begin{definition}
 \emph{[measurable selection]} A function $f:\mathcal{Y}\rightarrow\mathbb{R}^d$ is a measurable selection of a set valued map $F$ if 
 $f$ is measurable and for every $y\in\mathcal{Y}$, $f(y)\in F(y)$.
\end{definition}
For a set valued map $F$ let, $\mathscr{S}(F)$ denote the set of all measurable selections. The next lemma summarizes some standard
results about measurable set valued maps and their measurable selections.

\begin{lemma}\label{msel}
 For any measurable set valued map $F$,
 \begin{itemize}
  \item [(i)] $\mathscr{S}(F)\neq \emptyset$.
  \item [(ii)] \it{(Castaing representation) }\rm there exists $\left\{f_n\right\}_{n\geq1}\subseteq\mathscr{S}(F)$ such that, for every 
  $y\in\mathcal{Y}$, $F(y)=cl(\left\{f_n(y)\right\}_{n\geq1})$, where $cl(\cdot)$ denotes the closure of a set.
 \end{itemize}
\end{lemma}
We refer the reader to Theorem 1.2.6 and Theorem 1.2.7 in \cite{shoumei} for the proofs of Lemma \ref{msel}$(i)$ and $(ii)$ 
respectively.

\begin{definition}
 \emph{[$\mu$-integrable set-valued map]}
 Let $\mu$ be a probability measure on $(\mathcal{Y},\mathscr{F}_{\mathcal{Y}})$. A measurable set valued map $F$ is said to be 
 $\mu$-integrable if, there exists $f\in\mathscr{S}(F)$ which is $\mu$-integrable.
\end{definition}

\begin{definition}
 \emph{[Aumann's integral]} Let $\mu$ be a probability measure on $(\mathcal{Y},\mathscr{F}_{\mathcal{Y}})$. The integral of an 
 integrable set valued map $F$ is defined as,
 \begin{equation*}
  \int_{\mathcal{Y}} F(y)\mu(dy):=\left\{\int_{\mathcal{Y}} f(y)\mu(dy):\ f\in \mathscr{S}(F),\ f\ is\ \mu-integrable\right\}.
 \end{equation*}
\end{definition}

The next lemma states a useful result on the properties of the integral of a set valued map which is convex and compact set valued.
\begin{lemma}\label{cldint}
 Let $\mu$ be a probability measure on $(\mathcal{Y},\mathscr{F}_{\mathcal{Y}})$ and $F$ a $\mu$-integrable set valued map such that, 
 for every $y\in\mathcal{Y}$, $F(y)$ is convex and compact. Then, $\int_{\mathcal{Y}} F(y)\mu(dy)$ is a convex and closed subset of $\mathbb{R}^d$.
\end{lemma}
For a proof of the above lemma we refer the reader to Theorem 2.2.2 in \cite{shoumei}. The next lemma is a useful tool to prove 
measurability of set valued maps that arise later in this paper.

\begin{lemma}\label{mtool}
 Let $g:\mathcal{Y}\rightarrow\mathbb{R}^d$ be a measurable function and $F$ a measurable set valued map. Then the map 
 $y\rightarrow d(g(y),F(y))$ where $d(g(y),F(y)):=\inf_{z\in F(y)}\parallel g(y)-z \parallel$, is measurable.  
\end{lemma}
\begin{proof}
 Since $F$ is measurable, there exists $\left\{f_n\right\}_{n\geq1}$ a Castaing representation of $F$. Clearly for every $n\geq1$, the map
 $y\rightarrow\parallel g(y)-f_n(y)\parallel$ is measurable. Hence, the map $y\rightarrow\inf_{n\geq1}\parallel g(y)-f_n(y)\parallel$ 
 is measurable. Since for every $y\in\mathcal{Y}$, $cl(\left\{f_n(y)\right\}_{n\geq1})=F(y)$ we have, 
 $\inf_{n\geq1}\parallel g(y)-f_n(y)\parallel=d(g(y),F(y))$. Therefore the map $y\rightarrow d(g(y),F(y))$ is measurable.\qed
\end{proof}

\section{Limiting differential inclusion and its properties}
\label{ldi}

In this section we shall define the limiting differential inclusion (or the mean field) whose flow the iterates of recursion \eqref{rec} are 
expected to track. Analogous to the single valued case as in \cite{borkarmark}, the set valued map associated with the limiting differential 
inclusion is nothing but the set valued map obtained by averaging the set valued map $H$ in equation \eqref{rec} w.r.t. the stationary 
distributions of the Markov noise. The above shall be made precise in this section. Throughout this section we assume that $H$ denotes the 
set valued map satisfying $(A1)$ and $\left\{H^{(l)}\right\}_{l\geq1}$, $\left\{h^{(l)}\right\}_{l\geq1}$ are as in Lemma \ref{ctem} and Lemma \ref{param} 
respectively. 

First we define slices of the set valued maps $H$ and $H^{(l)}$ for every $x\in\mathbb{R}^d$.
\begin{definition}
 For every $x\in \mathbb{R}^d$, define,
 \begin{itemize}
  \item [(i)] $H_x:\mathcal{S}\rightarrow\left\{\text{subsets of }\mathbb{R}^d\right\}$ such that, for every $s\in\mathcal{S}$, $H_x(s):=H(x,s)$.
  \item [(ii)] for every $l\geq1$, $H^{(l)}_x:\mathcal{S}\rightarrow\left\{\text{subsets of }\mathbb{R}^d\right\}$ such that, for every 
  $s\in\mathcal{S}$, $\\H^{(l)}_x(s):=H^{(l)}(x,s)$.
  \item [(iii)] for every $l\geq1$, $h^{(l)}_x:\mathcal{S}\times U\rightarrow\mathbb{R}^d$ such that, for every 
  $(s,u)\in\mathcal{S}\times U$, $\\h^{(l)}_x(s,u):=h^{(l)}(x,s,u)$.
 \end{itemize}
\end{definition}

The lemma below summarizes properties that the set valued maps $H_x$ and $H_x^{(l)}$ inherit from the set valued maps $H$ and $H^{(l)}$ 
respectively. 
\begin{lemma}\label{msble}
 For every $x\in \mathbb{R}^d$, 
 \begin{itemize}
  \item [(i)] $H_x$ is a measurable set valued map and for every $s\in\mathcal{S}$, $H_x(s)$ is a convex and compact subset of 
  $\mathbb{R}^d$. Further there exists $C_x:=K(1+\parallel x\parallel)>0$, such that 
  $\\\sup_{s\in\mathcal{S}}\sup_{z\in H_x(s)}\parallel z\parallel\leq C_x$ (which means $H_x$ is bounded). 
  \item [(ii)] for every $l\geq1$, $H^{(l)}_x$ is a measurable set valued map and for every $s\in\mathcal{S}$, $H^{(l)}_x(s)$ 
  is a convex and compact subset of $\mathbb{R}^d$. Further for every $l\geq1$, there exists 
  $C^{(l)}_x:=K^{(l)}(1+\parallel x\parallel)>0$, such that 
  $\sup_{s\in\mathcal{S}}\sup_{z\in H^{(l)}_x(s)}\parallel z\parallel\leq C^{(l)}_x$ (which means $H^{(l)}_x$ is bounded). 
  \item [(iii)] For any probability measure $\mu$ on $(\mathcal{S},\mathscr{B}(\mathcal{S}))$, all measurable selections of $H_x$
  are $\mu$-integrable and hence $H_x$ is $\mu$-integrable.
  \item [(iv)] For any probability measure $\mu$ on $(\mathcal{S},\mathscr{B}(\mathcal{S}))$, all measurable selections of $H^{(l)}_x$
  are $\mu$-integrable and hence $H^{(l)}_x$ is $\mu$-integrable.
  \item [(v)] for every $l\geq1$, $h^{(l)}_x$ is continuous, bounded by the constant $C^{(l)}_x$ where $C^{(l)}_x$ is as in part 
  $(ii)$ of this lemma and for every $s\in\mathcal{S}$, $h^{(l)}_x(s,U)=H^{(l)}_x(s)$.
 \end{itemize}
\end{lemma}
\begin{proof} Fix some $x\in\mathbb{R}^d$.
 \begin{itemize}
  \item [(i)] For every $s\in\mathcal{S}$, the convexity and compactness of $H_x(s)$ follows from the definition of $H_x$ and 
  assumption $(A1)(i)$. Boundedness of $H_x$ follows from the definition of $H_x$ and assumption $(A1)(ii)$. From definition
  of $H_x$ and assumption $(A1)(iii)$, we deduce that for every sequence $\left\{s_n\right\}_{n\geq1}$ converging to $s$ and for every sequence 
  $\left\{z_n\in H_x(s_n)\right\}_{n\geq1}$ converging to $z$, $z\in H_x(s)$ (which means that the graph of $H_x$ is closed). 
  
  Let $C$ be some closed subset of $\mathbb{R}^d$. Consider a sequence $\left\{s_n\right\}_{n\geq1}$ converging to $s$ such that for every 
  $n\geq1$, $s_n\in \left\{s\in \mathcal{S}:H_x(s)\cap C\neq \emptyset\right\}$. For every $n\geq1$, let $z_n\in H_x(s_n)\cap C$. Since $H_x$ 
  is bounded, $\left\{z_n\right\}_{n\geq1}$ is a bounded sequence in $\mathbb{R}^d$ and hence has a convergent subsequence, say 
  $\left\{z_{n_k}\right\}_{k\geq1}$ converging to $z$. Since $C$ is closed and for every $k\geq1$, $z_{n_k}\in C$, we have $z\in C$. 
  Since the graph of $H_x$ is closed, we  have, $z\in H_x(s)$. Hence $H_x(s)\cap C\neq\emptyset$. Therefore 
  $\left\{s\in \mathcal{S}:H_x(s)\cap C\neq \emptyset\right\}$ is closed and clearly belongs to $\mathscr{B}(\mathcal{S})$, which establishes the 
  measurability of $H_x$.
  \item [(ii)] Fix some $l\geq1$. For every $s\in\mathcal{S}$, the convexity and compactness of $H^{(l)}_x(s)$ follows from the 
  definition of $H^{(l)}_x$ and Lemma \ref{ctem}$(i)$. Boundedness of $H^{(l)}_x$ follows from the definition of $H^{(l)}_x$ and 
  Lemma \ref{ctem}$(iii)$. Let $\left\{s_n\right\}_{n\geq1}$ be a sequence in $\mathcal{S}$ converging to $s$ and 
  $\left\{z_n\in H^{(l)}_x(s_n)\right\}_{n\geq1}$ converging to $z$. Using the continuity of $H^{(l)}$ and the definition of $H^{(l)}_x$, 
  we have, for every $\epsilon>0$, there exists $\delta>0$ such that 
  $\cup_{\left\{n:d(s_n,s)<\delta\right\}}H^{(l)}_x(s_n)\subseteq H^{(l)}_x(s)+\epsilon U$. Since $H^{(l)}_x(s)+\epsilon U$ is closed, we have,  
  $\overline{\cup_{\left\{n:d(s_n,s)<\delta\right\}}H^{(l)}_x(s_n)}\subseteq H^{(l)}_x(s)+\epsilon U$. Hence for every $\epsilon>0$, 
  $z\in H^{(l)}_x(s)+\epsilon U$. Since $H^{(l)}_x(s)$ is compact we have, $z\in H^{(l)}_x(s)$. Hence the graph of $H^{(l)}_x$ is 
  closed and by argument similar to the one in part $(i)$ of this proof we can prove the measurability of $H^{(l)}_x$.
  \item [(iii)] Since $H_x$ is bounded, for every $f\in\mathscr{S}(H_x)$, we have 
  $\sup_{s\in\mathcal{S}}\parallel f(s)\parallel\leq C_x$. Hence $\int_{\mathcal{S}}\parallel f(s)\parallel\mu(ds)\leq C_x$.
  \item [(iv)] Argument is same as in part $(iii)$ of this proof.
  \item [(v)] Follows from definition of $h^{(l)}_x$ and Lemma \ref{param}.\qed
 \end{itemize}
\end{proof}

Recall that $\mathcal{P}(\mathcal{S}\times U)$ denotes the set of probability measures on 
$(\mathcal{S}\times U,\mathscr{B}(\mathcal{S})\otimes\mathscr{B}(U))$ with the Prohorov's topology. 
For any probability measure $\nu\in\mathcal{P}(\mathcal{S}\times U)$, let $\nu_{\mathcal{S}}\in\mathcal{P}(\mathcal{S})$ denote the 
image of $\nu$ under the projection, $\mathcal{S}\times U\rightarrow\mathcal{S}$ (i.e. for every $A\in\mathscr{B}(\mathcal{S})$, 
$\nu_{\mathcal{S}}(A):=\int_{A\times U}\nu(ds,du)$). The next lemma provides a useful characterization of the integral of set valued maps 
$H^{(l)}_x$ in terms of their parametrization, $h^{(l)}_x$.

\begin{lemma}\label{chint}
 For every $x\in\mathbb{R}^d$, for every $l\geq1$, for every probability measure $\mu$ on $(\mathcal{S},\mathscr{B}(\mathcal{S}))$, 
 \begin{equation*}
  \int_{\mathcal{S}}H^{(l)}_x(s)\mu(ds)=\left\{\int_{\mathcal{S}\times U}h^{(l)}_x(s,u)\nu(ds,du): 
  \nu\in \mathcal{P}(\mathcal{S}\times\ U),\ \nu_{\mathcal{S}}=\mu\right\}
 \end{equation*}
\end{lemma}
\begin{proof}
 Fix some $x\in\mathbb{R}^d$, $l\geq1$ and $\mu\in\mathcal{P}(\mathcal{S})$.
 
 Let $z\in\int_{\mathcal{S}}H^{(l)}_x(s)\mu(ds)$. There exists $f\in\mathscr{S}(H^{(l)}_x)$ such that, $z=\int_{\mathcal{S}}f(s)\mu(ds)$. 
 Define a set valued map, $G:\mathcal{S}\rightarrow\left\{\text{subsets of }U\right\}$ such that, for every $s\in\mathcal{S}$, 
 $G(s):=\left\{u\in U: h^{(l)}(x,s,u)=f(s)\right\}$. Clearly for every $s\in\mathcal{S}$, $G(s)$ is non-empty and from the continuity of 
 $h^{(l)}$ we have that $G(s)$ is a closed subset of $U$. Moreover $G$ is a measurable set valued map (the proof of measurability has 
 been provided in appendix \ref{prf_msb}). Since $G$ is measurable, by Lemma \ref{msel}$(i)$ we have, $\mathscr{S}(G)\neq\emptyset$. For any 
 $g\in\mathscr{S}(G)$, we have, for every $s\in\mathcal{S}$, $h^{(l)}(x,s,g(s))=f(s)$. Define 
 $\hat{g}:\mathcal{S}\rightarrow\mathcal{S}\times U$ such that, for every $s\in\mathcal{S}$, $\hat{g}(s)=(s,g(s))$. Let 
 $\nu:=\mu\hat{g}^{-1}$ (push-forward measure). Clearly $\nu_{\mathcal{S}}=\mu$ and 
 $\int_{\mathcal{S}\times U}h^{(l)}(x,s,u)\nu(ds,du))=\int_{\mathcal{S}\times U}h^{(l)}(x,s,u)\mu\hat{g}^{-1}(ds,du))=
 \int_{\mathcal{S}}h^{(l)}(x,\hat{g}(s))\mu(ds)=\int_{\mathcal{S}}h^{(l)}(x,s,g(s))\mu(ds)=
 \int_{\mathcal{S}}f(s)\mu(ds)=z$. 
 
 Let $\nu\in\mathcal{P}(\mathcal{S}\times U)$ such that, $\nu_{\mathcal{S}}=\mu$. By Corollary 3.1.2 in \cite{borkarap}, there exists a 
 $\mu$-a.s. unique measurable map $q:\mathcal{S}\rightarrow\mathcal{P}(U)$ such that, $\nu(ds,du)=\mu(ds)q(s,du)$. By definition of 
 $h^{(l)}_x$ and Lemma \ref{param}, we have, for every $s\in\mathcal{S}$, $H^{(l)}_x(s)=h_x^{(l)}(s,U)=\left\{h_x^{(l)}(s,u):u\in U\right\}$ is 
 convex and compact subset of $\mathbb{R}^d$. Hence for every $s\in\mathcal{S}$, $\int_{U}h_x^{(l)}(s,u)q(s,du)\in H^{(l)}_x(s)$.
 Let $f:\mathcal{S}\rightarrow\mathbb{R}^d$ be such that for every $s\in\mathcal{S}$, $f(s)=\int_{U}h_x^{(l)}(s,u)q(s,du)$. Then,
 clearly $f$ is measurable and $f\in\mathscr{S}(H^{(l)}_x)$. Therefore, $\int_{\mathcal{S}\times U}h^{(l)}_x(s,u)\nu(ds,du)=
 \int_{\mathcal{S}}\big{[}\int_{U}h^{(l)}_x(s,u)q(s,du)\big{]}\mu(ds)=\int_{\mathcal{S}}f(s)\mu(ds)\in
 \int_{\mathcal{S}}H^{(l)}_x(s)\mu(ds)$.\qed
\end{proof}

Recall that for every $x\in\mathbb{R}^d$, $D(x)$ denotes the set of stationary distributions associated with the Markov chain whose 
transition probability kernel is given by $\Pi(x,\cdot)(\cdot)$ as in assumption $(A2)$. We shall now define the set valued maps associated 
with the limiting differential inclusion which is obtained by averaging the set valued map $H$ w.r.t. the stationary distributions of the 
Markov noise. We shall perform a similar operation on set valued maps $H^{(l)}$ which approximate $H$.

\begin{definition}:
 \begin{itemize}
  \item [(i)] Define $\hat{H}:\mathbb{R}^d\rightarrow\left\{\text{subsets of }\mathbb{R}^d\right\}$ such that, for every $x\in\mathbb{R}^d$,
  \begin{equation}\label{lim1}
  \hat{H}(x):=\cup_{\mu\in D(x)}\int_{\mathcal{S}}H_x(s)\mu(ds).
  \end{equation}
  \item [(ii)] For every $l\geq1$, define $\hat{H}^{(l)}:\mathbb{R}^d\rightarrow\left\{\text{subsets of }\mathbb{R}^d\right\}$ such that, for every 
  $x\in\mathbb{R}^d$, 
  \begin{equation}\label{lim2}
  \hat{H}^{(l)}(x):=\cup_{\mu\in D(x)}\int_{\mathcal{S}}H^{(l)}_x(s)\mu(ds).   
  \end{equation}
 \end{itemize}
\end{definition}

For every $l\geq1$, the differential inclusion (DI) associated with the set valued map $\hat{H}^{(l)}$ is given by,
\begin{equation}\label{di1}
 \frac{dx}{dt}\in \hat{H}^{(l)}(x).
\end{equation}
Similarly the differential inclusion (DI) associated with the set valued map $\hat{H}$ is given by,
\begin{equation}\label{di2}
 \frac{dx}{dt}\in \hat{H}(x).
\end{equation}

A function $\bf{x}\rm:\mathbb{R}\rightarrow\mathbb{R}^d$ is said to be a \bf{solution }\rm of DI \eqref{di1} (or DI \eqref{di2}) with initial condition
$x_0\in\mathbb{R}^d$ if, $\bf{x}\rm$ is absolutely continuous, $\bf{x}\rm(0)=x_0$ and for a.e. $t\in\mathbb{R}$, 
$\frac{d\bf{x}\rm(t)}{dt}\in\hat{H}^{(l)}(\bf{x}\rm(t))$ (or $\frac{d\bf{x}\rm(t)}{dt}\in\hat{H}(\bf{x}\rm(t))$).

For every $l\geq1$, for every initial condition $x_0\in\mathbb{R}^d$, let $\Sigma^{(l)}(x_0)$ denote the set of solutions to DI \eqref{di1} 
with initial condition $x_0$, that is, 
\begin{equation*}
 \Sigma^{(l)}(x_0)=\left\{\bf{x}\rm:\mathbb{R}\rightarrow\mathbb{R}^d: \bf{x}\rm\ is\ a\ solution\ of\ DI\ \eqref{di1},\ \bf{x}\rm(0)=x_0\right\}.
\end{equation*}
Similarly, let $\Sigma(x_0)$ denote the set of solutions of DI \eqref{di2} with initial condition $x_0\in\mathbb{R}^d$. For any 
$A\subseteq\mathbb{R}^d$, for every $l\geq1$, $\Sigma^{(l)}(A):=\cup_{x_0\in A}\Sigma^{(l)}(x_0)$ (similarly, 
$\Sigma(A):=\cup_{x_0\in A}\Sigma(x_0)$). 

The next lemma establishes some properties of the set valued map $\hat{H}^{(l)}$ which ensure the existence of solutions to DI \eqref{di1}. 
\begin{lemma}\label{march1}
 For every $l\geq1$, $\hat{H}^{(l)}$ satisfies the following,
 \begin{itemize}
  \item [(i)] for $K^{(l)}$ as in Lemma \ref{ctem}$(iii)$, we have, for every $x\in\mathbb{R}^d$, 
  $\\\sup_{z\in\hat{H}^{(l)}(x)}\parallel z\parallel\leq K^{(l)}(1+\parallel x\parallel)$.
  \item [(ii)] for every $x\in\mathbb{R}^d$, $\hat{H}^{(l)}(x)$ is a non-empty, convex and compact subset of $\mathbb{R}^d$.
  \item [(iii)] for every sequence $\left\{x_n\right\}_{n\geq1}$ in $\mathbb{R}^d$ converging to $x$, for every sequence 
  $\left\{z_n\in\hat{H}^{(l)}(x_n)\right\}_{n\geq1}$ converging to $z$, we have, $z\in\hat{H}^{(l)}(x)$.
 \end{itemize}
\end{lemma}
\begin{proof}
 Fix $l\geq1$.
 \begin{itemize}
  \item [(i)] For every $x\in\mathbb{R}^d$, from Lemma \ref{msble}$(ii)$, we have, for every $f\in\mathscr{S}(H^{(l)}_x)$, for every 
  $\mu\in D(x)$, $\parallel\int_{\mathcal{S}}f(s)\mu(ds)\parallel\leq\int_{\mathcal{S}}\parallel f(s)\parallel\mu(ds)\leq C^{(l)}_x=
  K^{(l)}(1+\parallel x\parallel)$. Using the above and from the definition of $\hat{H}^{(l)}$ we have, 
  $\sup_{z\in \hat{H}^{(l)}(x)}\parallel z\parallel=\\\sup_{\mu\in D(x)}\sup_{f\in\mathscr{S}(H^{(l)}_x)}\parallel 
  \int_{\mathcal{S}}f(s)\mu(ds)\parallel\leq K^{(l)}(1+\parallel x\parallel)$.
  \item [(ii)] For every $x\in\mathbb{R}^d$, $\hat{H}^{(l)}(x)$ is clearly non-empty. Let $z_1,z_2\in\hat{H}^{(l)}(x)$ and 
  $\theta\in(0,1)$. By definition of $\hat{H}^{(l)}$ and by Lemma \ref{chint}, we have that there exist $\nu^{1},\nu^{2}\in 
  \mathcal{P}(\mathcal{S}\times U)$ such that, for $i\in\left\{1,2\right\}$, $z_i=\int_{\mathcal{S}\times U}h^{(l)}_x(s,u)\nu^i(ds,du)$ and 
  $\nu^{i}_{\mathcal{S}}\in D(x)$. Hence,
  \begin{align*}
   \theta z_1+(1-\theta)z_2&=\theta\int_{\mathcal{S}\times U}h^{(l)}_x(s,u)\nu^1(ds,du)+
   (1-\theta)\int_{\mathcal{S}\times U}h^{(l)}_x(s,u)\nu^2(ds,du)\\
                           &=\int_{\mathcal{S}\times U}h^{(l)}_x(s,u)(\theta\nu^1+(1-\theta)\nu^2)(ds,du),
  \end{align*}
  where, $\theta\nu^1+(1-\theta)\nu^2\in\mathcal{P}(\mathcal{S}\times U)$ such that, for every  
  $A\in\mathscr{B}(\mathcal{S})\otimes \mathscr{B}(U)$, $\\(\theta\nu^1+(1-\theta)\nu^2)(A)=\theta\nu^1(A)+(1-\theta)\nu^2(A)$. Further 
  $(\theta\nu^1+(1-\theta)\nu^2)_{\mathcal{S}}=\theta\nu^1_{\mathcal{S}}+(1-\theta)\nu^2_{\mathcal{S}}$ and it belongs to 
  $D(x)$, since for every $x\in\mathbb{R}^d$, $D(x)$ is convex.
  
  By part $(i)$ of this proof, for every $x\in\mathbb{R}^d$, $\hat{H}^{(l)}(x)$ is bounded. Hence in order to show compactness it is enough 
  to prove that for every $x\in\mathbb{R}^d$, $\hat{H}^{(l)}(x)$ is closed. Let $\left\{z_n\right\}_{n\geq1}$ be a sequence in $\hat{H}^{(l)}(x)$ 
  for some $x\in\mathbb{R}^d$, converging to $z$. By Lemma \ref{chint}, we know that for every $n\geq1$, there exists 
  $\nu^n\in\mathcal{P}(\mathcal{S}\times U)$ such that, $z_n=\int_{\mathcal{S}\times U}h^{(l)}_x(s,u)\nu^n(ds,du)$ and 
  $\nu^n_{\mathcal{S}}\in D(x)$. Since $\mathcal{S}\times U$ is a compact metric space, by Prohorov's theorem (see theorem 2.3.1 in 
  \cite{borkarap}), there exists a subsequence $\left\{n_k\right\}_{k\geq1}$ such that, $\left\{\nu^{n_k}\right\}_{k\geq1}$ converges to some 
  $\nu\in\mathcal{P}(\mathcal{S}\times U)$. Clearly $\left\{z_{n_k}\right\}_{k\geq1}$ converges to $z$ and $\left\{\nu^{n_k}_{\mathcal{S}}\right\}_{k\geq1}$ 
  converges to $\nu_{\mathcal{S}}$ in $\mathcal{P}(\mathcal{S})$. Since for every $x\in\mathbb{R}^d$, $D(x)$ is a closed subset of 
  $\mathcal{P}(\mathcal{S})$, we have $\nu_{\mathcal{S}}\in D(x)$. By theorem 2.1.1$(ii)$ in \cite{borkarap} and continuity of 
  $h^{(l)}_x$ we have $\left\{\int_{\mathcal{S}\times U}h^{(l)}_x(s,u)\nu^{n_k}(ds,du)\right\}_{k\geq1}$ converges to 
  $\int_{\mathcal{S}\times U}h^{(l)}_x(s,u)\nu(ds,du)$ and hence $z=\int_{\mathcal{S}\times U}h^{(l)}_x(s,u)\nu(ds,du)$ with 
  $\nu_{\mathcal{S}}\in D(x)$. Therefore $z\in \hat{H}^{(l)}(x)$ and hence $\hat{H}^{(l)}(x)$ is closed.
  \item [(iii)] Let $\left\{x_n\right\}_{n\geq1}$ be a sequence in $\mathbb{R}^d$ converging to $x$ and $\left\{z_n\in \hat{H}^{(l)}(x_n)\right\}_{n\geq1}$ 
  converging to $z$. By Lemma \ref{chint}, we have, for every $n\geq1$, there exists $\nu^n\in\mathcal{P}(\mathcal{S}\times U)$,
  such that $z_n=\int_{\mathcal{S}\times U}h_{x_n}^{(l)}(s,u)\nu^n(ds,du)$ and $\nu^n_{\mathcal{S}}\in D(x_n)$. Since 
  $\mathcal{S} \times U$ is a compact metric space, by Prohorov's theorem we have that, there exists a subsequence, $\left\{n_k\right\}_{k\geq1}$
  such that, $\left\{\nu^{n_k}\right\}_{k\geq1}$ converges to some $\nu\in\mathcal{P}(\mathcal{S}\times U)$. Since $\left\{x_{n_k}\right\}_{k\geq1}$ converges 
  to $x$, we have $\left\{h^{(l)}_{x_{n_k}}\right\}_{k\geq1}$ converges uniformly to $h^{(l)}_x$ (for a proof see appendix \ref{prf_conv}). By uniform convergence 
  of $h^{(l)}_{x_{n_k}}$ to $h^{(l)}_x$ and Theorem 2.1.1$(ii)$ in \cite{borkarap} we have that the sequence 
  $\left\{\int_{\mathcal{S}\times U}h^{(l)}_{x_{n_k}}(s,u)\nu^{n_k}(ds,du)\right\}_{k\geq1}$ converges to 
  $\int_{\mathcal{S}\times U}h^{(l)}_x(s,u)\nu(ds,du)$. Hence $z=\int_{\mathcal{S}\times U}h^{(l)}_x(s,u)\nu(ds,du)$. Clearly 
  $\left\{\nu^{n_k}_{\mathcal{S}}\right\}_{k\geq1}$ converges to $\nu_{\mathcal{S}}$ and by closed graph property of the map $x\rightarrow D(x)$, we 
  have $\nu_{\mathcal{S}}\in D(x)$. Hence $z=\int_{\mathcal{S}\times U}h^{(l)}_x(s,u)\nu(ds,du)$ and $\nu_{\mathcal{S}}\in D(x)$. 
  Therefore by Lemma \ref{chint} and definition of $\hat{H}^{(l)}$, we have, $z\in \hat{H}^{(l)}(x)$.\qed
 \end{itemize}
\end{proof}
The set valued map satisfying properties stated in Lemma \ref{march1}, is called a Marchaud map (see page 62 of \cite{aubinvt}).
By Lemma \ref{march1}, we know that for every $l\geq1$, the set valued maps $\hat{H}^{(l)}$ are Marchaud maps. For such maps their 
associated DIs are known to admit at least one solution through every initial condition (see \cite{aubindi}, chapter 2.1 or 
\cite{benaim1}, section 1.2). Thus for every $l\geq1$, for every initial condition $x_0\in\mathbb{R}^d$, $\Sigma^{(l)}(x_0)\neq\emptyset$.

The next lemma establishes relation between $\hat{H}$ and $\hat{H}^{(l)}$.  
\begin{lemma}\label{approx1}
 For every $x\in\mathbb{R}^d$,
 \begin{itemize}
  \item [(i)] for every $l\geq1$, $\hat{H}(x)\subseteq\hat{H}^{(l+1)}(x)\subseteq\hat{H}^{(l)}(x)$.
  \item [(ii)]$\cup_{\mu\in D(x)}\cap_{l\geq1}\int_{\mathcal{S}}H_x^{(l)}(s)\mu(ds)=\cap_{l\geq1}\hat{H}^{(l)}(x)$.
  \item [(iii)] $\hat{H}(x)=\cap_{l\geq1}\hat{H}^{(l)}(x)$.
 \end{itemize}
\end{lemma}
\begin{proof}
 Fix $x\in\mathbb{R}^d$.
    \begin{itemize}
      \item [(i)] Fix $l\geq1$.
      By Lemma \ref{ctem}$(ii)$, for every $(x,s)\in\mathbb{R}^d\times\mathcal{S}$, $H(x,s)\subseteq H^{(l+1)}(x,s)\subseteq 
      H^{(l)}(x,s)$. Hence, for every $x\in\mathbb{R}^d$, by definition of $H_x$ and $H^{(l)}_x$ we have, for every $s\in\mathcal{S}$,
      $H_x(s)\subseteq H^{(l+1)}_x(s)\subseteq H^{(l)}_x(s)$. Therefore, $\mathscr{S}(H_x)\subseteq\mathscr{S}(H_x^{(l+1)})\subseteq
      \mathscr{S}(H^{(l)}_x)$. From the above we obtain that, for every $\mu\in D(x)$, $\int_{\mathcal{S}}H_x(s)\mu(ds)\subseteq
      \int_{\mathcal{S}}H_x^{(l+1)}(s)\mu(ds)\subseteq\int_{\mathcal{S}}H^{(l)}_x(s)\mu(ds)$ which leads to the desired conclusion.
      \item [(ii)] Clearly by definition of the set valued map $\hat{H}^{(l)}$ we have that 
      $\cup_{\mu\in D(x)}\cap_{l\geq1}\int_{\mathcal{S}}H_x^{(l)}(s)\mu(ds)\subseteq\cap_{l\geq1}\hat{H}^{(l)}(x)$. Let 
      $z\in\cap_{l\geq1}\hat{H}^{(l)}(x)$. Then for every $l\geq1$, there exists $\mu^{(l)}\in D(x)$ such that $z\in\int_{\mathcal{S}}
      H^{(l)}_x(s)\mu^{(l)}(ds)$. By Prohorov's theorem the sequence of probability measures $\left\{\mu^{(l)}\right\}_{l\geq1}$ is relatively 
      compact in $\mathcal{P}(\mathcal{S})$ and hence has a limit point, say $\mu$. Let $\left\{\mu^{(l_k)}\right\}_{k\geq1}$ denote a subsequence 
      with $\mu$ as its limit and since $D(x)$ is a closed subset of $\mathcal{P}(\mathcal{S})$, we have that $\mu\in D(x)$. For every 
      $l\geq1$, for every $k$ such that $l_k\geq l$, we have that $\mathscr{S}(H_x^{(l_k)})\subseteq\mathscr{S}(H_x^{(l)})$. Thus for 
      every $l\geq1$, for every $k$ such that $l_k\geq l$, we have that $z\in\int_{\mathcal{S}}H^{(l)}_x(s)\mu^{(l_k)}(ds)$. By Lemma 
      \ref{chint}, we have that for every $l\geq1$, for every $k$ such that $l_k\geq l$, there exists $\nu^{(l,l_k)}\in
      \mathcal{P}(\mathcal{S}\times U)$ with $\nu^{(l,l_k)}_{\mathcal{S}}=\mu^{(l_k)}$ and $z=\int_{\mathcal{S}\times U}h^{(l)}(x,s,u)
      \nu^{(l,l_k)}(ds,du)$. 
      
      Fix $l\geq1$. Since $\mathcal{S}\times U$ is a compact metric space, by Prohorov's theorem we have that the 
      sequence $\left\{\nu^{(l,l_k)}\right\}_{k\geq1}$ is relatively compact in $\mathcal{P}(\mathcal{S}\times U)$. Let $\nu^{(l)}$ denote a limit 
      point of the above sequence and its clear that $\nu^{(l)}_{\mathcal{S}}=\mu$. Then clearly $z=\int_{\mathcal{S}\times U}
      h^{(l)}(x,s,u)\nu^{(l)}(ds,du)$ and hence by Lemma \ref{chint}, $z\in\int_{\mathcal{S}}H^{(l)}_x(s)\mu(ds)$.
      Since $l\geq1$ is arbitrary, we get that for $\mu\in D(x)$, for every $l\geq1$, $z\in\int_{\mathcal{S}}H^{(l)}_x(s)\mu(ds)$.
      Thus $z\in\cup_{\mu\in D(x)}\cap_{l\geq1}\int_{\mathcal{S}}H_x^{(l)}(s)\mu(ds)$ and hence $\cap_{l\geq1}\hat{H}^{(l)}(x)\subseteq
      \cup_{\mu\in D(x)}\cap_{l\geq1}\int_{\mathcal{S}}H_x^{(l)}(s)\mu(ds)$.
      \item [(iii)] By part$(i)$ of this lemma we have, for every $x\in\mathbb{R}^d$, $\hat{H}(x)\subseteq\cap_{l\geq1}\hat{H}^{(l)}(x)$. 
      Therefore its enough to show that for every $x\in\mathbb{R}^d$, $\cap_{l\geq1}\hat{H}^{(l)}(x)\subseteq\hat{H}(x)$. Fix 
      $x\in\mathbb{R}^d$ and $\mu\in D(x)$. Let $z\in\cap_{l\geq1}\int_{\mathcal{S}}H^{(l)}_x(s)\mu(ds)$. For every $l\geq1$, there 
      exists $f^{(l)}\in\mathscr{S}(H^{(l)}_x)$ such that $z=\int_{\mathcal{S}}f^{(l)}(s)\mu(ds)$. Let 
      $d(z,\int_{\mathcal{S}}H_x(s)\mu(ds)):=\inf\left\{\parallel z-y\parallel: y\in\int_{\mathcal{S}}H_x(s)\mu(ds)\right\}$. Then, for every 
      $l\geq1$,
      \begin{align*}
       d(z,\int_{\mathcal{S}}H_x(s)\mu(ds))&=\inf_{f\in\mathscr{S}(H_x)}\parallel\int_{\mathcal{S}}(f^{(l)}(s)-f(s))\mu(ds)\parallel\\
                                           &\leq\inf_{f\in\mathscr{S}(H_x)}\int_{\mathcal{S}}\parallel f^{(l)}(s)-f(s)\parallel\mu(ds)\\
                                           &=\int_{\mathcal{S}}\inf_{f\in\mathscr{S}(H_x)}\parallel f^{(l)}(s)-f(s)\parallel\mu(ds),
      \end{align*}
      where the last equality follows from Lemma 1.3.12 in \cite{shoumei}. By Lemma \ref{msble}$(i)$, $H_x$ is a measurable set valued 
      map and hence by Lemma \ref{msel}$(ii)$, $H_x$ admits a Castaing representation, $\left\{f_n\right\}_{n\geq1}$. Hence for every 
      $s\in\mathcal{S}$, $d(f^{(l)}(s),H_x(s))\leq\inf_{f\in\mathscr{S}(H_x)}\parallel f^{(l)}(s)-f(s)\parallel\leq
      \inf_{n\geq1}\parallel f^{(l)}(s)-f_n(s)\parallel=d(f^{(l)}(s),H_x(s))$ where $s\rightarrow d(f^{(l)}(s),H_x(s))$ is as in Lemma 
      \ref{mtool}. Therefore, for every $l\geq1$,
      \begin{equation*}
       d(z,\int_{\mathcal{S}}H_x(s)\mu(ds))\leq\int_{\mathcal{S}}d(f^{(l)}(s),H_x(s))\mu(ds).
      \end{equation*}      
      By observation (b) stated after Lemma \ref{ctem}, we have that, for every $s\in\mathcal{S}$, 
      $\\\lim_{l\to\infty}d(f^{(l)}(s),H_x(s))=0$. Further by observation (a) stated after Lemma \ref{ctem}, Lemma \ref{msble}$(i)-(ii)$, 
      we have that, for every $s\in\mathcal{S}$, $d(f^{(l)}(s),H_x(s))\leq(\tilde{K}+K)(1+\parallel x\parallel)$. Thus, by 
      bounded convergence theorem we get, $\lim_{l\to\infty}\int_{\mathcal{S}}d(f^{(l)}(s),H_x(s))\mu(ds)=0$. Hence, 
      $d(z,\int_{\mathcal{S}}H_x(s)\mu(ds))=0$. By Lemma \ref{cldint}, we have that $\int_{\mathcal{S}}H_x(s)\mu(ds)$ is closed. 
      Therefore $z\in\int_{\mathcal{S}}H_x(s)\mu(ds)$. 
      
      From the arguments in the previous paragraph, we have that for every $x\in\mathbb{R}^d$, for every $\mu\in D(x)$, 
      $\cap_{l\geq1}\int_{\mathcal{S}}H^{(l)}_x(s)\mu(ds)\subseteq\int_{\mathcal{S}}H_x(s)\mu(ds)$. Thus, for every $x\in\mathbb{R}^d$, 
      $\\\cup_{\mu\in D(x)}\cap_{l\geq1}\int_{\mathcal{S}}H^{(l)}_x(s)\mu(ds)\subseteq\cup_{\mu\in D(x)}\int_{\mathcal{S}}H_x(s)\mu(ds)$. 
      Therefore, by part $(ii)$ of this lemma we have that for every $x\in\mathbb{R}^d$, 
      $\cap_{l\geq1}\hat{H}^{(l)}(x)\subseteq\hat{H}(x)$.\qed
    \end{itemize}
\end{proof}

In the next lemma we show that the set valued map $\hat{H}$ is a Marchaud map.
\begin{lemma}\label{march2}
 The set valued map $\hat{H}$ satisfies the following:
 \begin{itemize}
  \item [(i)] For $K>0$ as in Lemma \ref{msble}$(i)$, for every $x\in\mathbb{R}^d$, 
  $\sup_{z\in\hat{H}(x)}\parallel z\parallel\leq K(1+\parallel x\parallel)$.
  \item [(ii)] For every $x\in\mathbb{R}^d$, $\hat{H}(x)$  is a non-empty, convex and compact subset of $\mathbb{R}^d$.
  \item [(iii)] For every sequence $\left\{x_n\right\}_{n\geq1}$ in $\mathbb{R}^d$ converging to $x$, for every sequence 
  $\left\{z_n\in\hat{H}(x_n)\right\}_{n\geq1}$ converging to $z$, we have, $z\in\hat{H}(x)$.
 \end{itemize}
\end{lemma}
\begin{proof}:
 \begin{itemize}
  \item [(i)] For every $x\in\mathbb{R}^d$, from Lemma \ref{msble}$(i)$, we have, for every $f\in\mathscr{S}(H_x)$, for every 
  $\mu\in D(x)$, $\parallel\int_{\mathcal{S}}f(s)\mu(ds)\parallel\leq\int_{\mathcal{S}}\parallel f(s)\parallel\mu(ds)\leq C_x=
  K(1+\parallel x\parallel)$. Using the above and from the definition of $\hat{H}$ we have, 
  $\sup_{z\in \hat{H}(x)}\parallel z\parallel=\sup_{\mu\in D(x)}\sup_{f\in\mathscr{S}(H_x)}\parallel
  \int_{\mathcal{S}}f(s)\mu(ds)\parallel\\\leq K(1+\parallel x\parallel)$.
  \item [(ii)] For every $x\in\mathbb{R}^d$, by Lemma \ref{msble}$(iii)$, for every $\mu\in D(x)$, $H_x$ is $\mu$-integrable. Hence 
  $\hat{H}_x$ is non-empty. By Lemma \ref{march1}$(ii)$ and Lemma \ref{approx1}$(iii)$, we have that, $\hat{H}(x)$ is a convex and 
  compact subset of $\mathbb{R}^d$.
  \item [(iii)] Let $\left\{x_n\right\}_{n\geq1}$ be a sequence in $\mathbb{R}^d$, converging to $x$ and $\left\{z_n\in \hat{H}(x_n)\right\}_{n\geq1}$ 
  converging to $z$. By Lemma \ref{approx1}$(i)$, for every $l\geq1$, for every $n\geq1$, $z_n\in\hat{H}^{(l)}(x_n)$. Hence by Lemma 
  \ref{march1}$(iii)$, we have that, for every $l\geq1$, $z\in\hat{H}^{(l)}(x)$. Therefore by Lemma \ref{approx1}$(iii)$, we have, 
  $z\in\hat{H}(x)$.\qed 
 \end{itemize}
\end{proof}

By Lemma \ref{march2}, we know that the set valued map $\hat{H}$ is a Marchaud map. Hence for every $x_0\in\mathbb{R}^d$, 
$\Sigma(x_0)\neq\emptyset$.

By definition, $\Sigma^{(l)}(x_0)$ and $\Sigma(x_0)$ are subsets of $\mathcal{C}(\mathbb{R},\mathbb{R}^d)$, the space of all $\mathbb{R}^d$ 
valued continuous functions on $\mathbb{R}$. $\mathcal{C}(\mathbb{R},\mathbb{R}^d)$ is a complete metric space for distance $\bf{D}\rm$, 
defined by,
\begin{equation*}
 \bf{D}\rm(\bf{x}\rm,\bf{z}\rm):=\sum_{k=1}^{\infty}\frac{1}{2^k}\min(\parallel \bf{x}\rm-\bf{z}\rm\parallel_{[-k,k]},1),
\end{equation*}
where $\parallel\cdot\parallel_{[-k,k]}$ denotes the sup norm on $\mathcal{C}([-k,k],\mathbb{R}^d)$.

The next lemma summarizes some properties of the solutions of DI \eqref{di1},\eqref{di2} and the relation between them.

\begin{lemma}\hfill
\begin{itemize}
 \item [(i)] For every $A\subseteq\mathbb{R}^d$ compact, $\Sigma(A)$ and for every $l\geq1$, $\Sigma^{(l)}(A)$ are compact subsets of 
 $\mathcal{C}(\mathbb{R},\mathbb{R}^d)$.
 \item [(ii)] For every $x_0\in\mathbb{R}^d$, for every $l\geq1$, $\Sigma(x_0)\subseteq\Sigma^{(l+1)}(x_0)\subseteq\Sigma^{(l)}(x_0)$.
 \item [(iii)] For every $x_0\in\mathbb{R}^d$, $\Sigma(x_0)=\cap_{l\geq1}\Sigma^{(l)}(x_0)$.
\end{itemize}
\end{lemma}
\begin{proof}:
 \begin{itemize}
  \item [(i)] Follows from Lemma 3.1 in \cite{benaim1}. 
  \item [(ii)] Fix some $x_0\in\mathbb{R}^d$ and $l\geq1$. Let $\bf{x}\rm\in\Sigma(x_0)$. Then for $a.e.\ t\in\mathbb{R}$, 
  $\frac{d\bf{x}\rm(t)}{dt}\in\hat{H}(\bf{x}\rm(t))$. By Lemma \ref{approx1}$(i)$, we have, for every $t\geq0$, 
  $\hat{H}(\bf{x}\rm(t))\subseteq\hat{H}^{(l)}(\bf{x}\rm(t))$. Therefore, for $a.e.\ t\in\mathbb{R}$, 
  $\frac{d\bf{x}\rm(t)}{dt}\in\hat{H}^{(l)}(\bf{x}\rm(t))$. Hence, $\bf{x}\rm\in\Sigma^{(l)}(x_0)$. A similar argument gives us that, 
  $\Sigma^{(l+1)}(x_0)\subseteq\Sigma^{(l)}(x_0)$.
  \item [(iii)] By part $(ii)$ of this lemma we have that, for every $x_0\in\mathbb{R}^d$, 
  $\Sigma(x_0)\subseteq\cap_{l\geq1}\Sigma^{(l)}(x_0)$. Therefore, it is enough to show that, for every $x_0\in\mathbb{R}^d$, 
  $\cap_{l\geq1}\Sigma^{(l)}(x_0)\subseteq\Sigma(x_0)$. Let $\bf{x}\rm\in\cap_{l\geq1}\Sigma^{(l)}(x_0)$ for some $x_0\in\mathbb{R}^d$. 
  Then for every $l\geq1$, for $a.e.\ t\in\mathbb{R}$, $\frac{d\bf{x}\rm(t)}{dt}\in\hat{H}^{(l)}(\bf{x}\rm(t))$. Thus for 
  $a.e.\ t\in\mathbb{R}$, $\frac{d\bf{x}\rm(t)}{dt}\in\cap_{l\geq1}\hat{H}^{(l)}(\bf{x}\rm(t))$. By Lemma \ref{approx1}$(iii)$, we have 
  that, for $a.e.\ t\in\mathbb{R}$, $\frac{d\bf{x}\rm(t)}{dt}\in\hat{H}(\bf{x}\rm(t))$. Hence $\bf{x}\rm\in\Sigma(x_0)$.\qed
 \end{itemize}
\end{proof}

\section{Space of probability measure valued functions}
\label{sopmvf}
In this section we define the space of probability measure valued measurable functions on $[0,\infty)$ and introduce an appropriate 
topology on this space. Such spaces are used in the theory of existence of optimal control for diffusions and can also be found in 
\cite{borkarmark} and \cite{borkaropt}. We shall use the following in the analysis of recursion \eqref{rec} presented in the next section.

Let $\mathcal{M}(\mathcal{S}\times U)$ denote the set of all functions $\gamma:[0,\infty)\rightarrow\mathcal{P}(\mathcal{S}\times U)$, 
measurable. Formally,
\begin{equation*}
 \mathcal{M}(\mathcal{S}\times U):=\left\{\gamma:[0,\infty)\rightarrow\mathcal{P}(\mathcal{S}\times U):\ \gamma\ is\ measurable\right\}.
\end{equation*}
Similarly, define $\mathcal{M}(\mathcal{S})$, the set of all functions $\gamma:[0,\infty)\rightarrow\mathcal{P}(\mathcal{S})$, measurable.
Formally, 
\begin{equation*}
 \mathcal{M}(\mathcal{S}):=\left\{\gamma:[0,\infty)\rightarrow\mathcal{P}(\mathcal{S}):\ \gamma\ is\ measurable\right\}.
\end{equation*}

Let $\tau_{\mathcal{S}\times U}$ denote the topology on $\mathcal{M}(\mathcal{S}\times U)$ which is the coarsest topology that 
renders continuous the maps $\gamma\in\mathcal{M}(\mathcal{S}\times U)\rightarrow\int_{0}^{T}g(t)[\int_{\mathcal{S}\times U}f(s,u)
\gamma(t)(ds,du)]dt\in\mathbb{R}$ for all $f\in\mathcal{C}(\mathcal{S}\times U,\mathbb{R})$, for all $T>0$ and for all 
$g\in\mathbb{L}_2[0,T]$.

Similarly let $\tau_{\mathcal{S}}$ denote the topology on $\mathcal{M}(\mathcal{S})$ which is the coarsest topology that renders 
continuous the maps $\gamma\in\mathcal{M}(\mathcal{S})\rightarrow\int_{0}^{T}g(t)[\int_{\mathcal{S}}f(s)\gamma(t)(ds)]dt\in\mathbb{R}$ 
for all $f\in\mathcal{C}(\mathcal{S},\mathbb{R})$, for all $T>0$ and for all $g\in\mathbb{L}_2[0,T]$.

The next result is a well known metrization lemma for the above defined topological spaces. 

\begin{lemma}
 \emph{[Metrization]} 
 \begin{itemize}
  \item [(i)] The topological space $(\mathcal{M}(\mathcal{S}\times U),\tau_{\mathcal{S}\times U})$ is compact metrizable.
  \item [(ii)] The topological space $(\mathcal{M}(\mathcal{S}),\tau_{\mathcal{S}})$ is compact metrizable.
 \end{itemize}
\end{lemma}
For a proof of the above lemma we refer the reader to Lemma 2.1 in \cite{borkarmark}. The next lemma provides a continuous map from 
$\mathcal{M}(\mathcal{S}\times U)$ to $\mathcal{M}(\mathcal{S})$ which is used later in this paper. Recall that for any probability 
measure, $\nu\in\mathcal{P}(\mathcal{S}\times U)$, $\nu_{\mathcal{S}}$ denotes the image of $\nu$ under projection, 
$\mathcal{S}\times U\rightarrow\mathcal{S}$ (i.e. $\nu_{\mathcal{S}}\in\mathcal{P}(\mathcal{S})$ such that for every 
$A\in\mathscr{B}(\mathcal{S})$, $\nu_{\mathcal{S}}(A)=\int_{A\times U}\nu(ds,du)$).

\begin{lemma}:\label{lambd}
 \begin{itemize}
  \item [(i)]The map $\lambda:\mathcal{P}(\mathcal{S}\times U)\rightarrow\mathcal{P}(\mathcal{S})$ such that, for every 
  $\nu\in\mathcal{P}(\mathcal{S}\times U)$, $\lambda(\nu):=\nu_{\mathcal{S}}$, is continuous.
  \item [(ii)] For every $\gamma\in\mathcal{M}(\mathcal{S}\times U)$, $\lambda\circ\gamma\in\mathcal{M}(\mathcal{S})$.
  \item [(iii)] The map $\Lambda:\mathcal{M}(\mathcal{S}\times U)\rightarrow\mathcal{M}(\mathcal{S})$ such that, for every 
  $\gamma\in\mathcal{M}(\mathcal{S}\times U)$, $\Lambda(\gamma):=\lambda\circ\gamma$, is continuous.
 \end{itemize}
\end{lemma}
\begin{proof}:
 \begin{itemize}
  \item [(i)] Let $\left\{\nu^n\right\}_{n\geq1}$ be a sequence in $\mathcal{P}(\mathcal{S}\times U)$ converging to 
  $\nu\in\mathcal{P}(\mathcal{S}\times U)$. Let $\pi:\mathcal{S}\times U\rightarrow\mathcal{S}$ denote the projection map. 
  Then for every $f\in\mathcal{C}(\mathcal{S},\mathbb{R})$, $f\circ\pi\in\mathcal{C}(\mathcal{S}\times U,\mathbb{R})$. By Theorem 2.1.1 
  $(ii)$ in \cite{borkarap} and the compactness of $\mathcal{S}\times U$, for every $f\in\mathcal{C}(\mathcal{S},\mathbb{R})$, 
  $\int_{\mathcal{S}\times U}\left(f\circ\pi\right)(s,u)\nu^n(ds,du)\to\int_{\mathcal{S}\times U}\left(f\circ\pi\right)(s,u)\nu(ds,du)$ as $n\to\infty$. Hence for 
  every $f\in\mathcal{C}(\mathcal{S},\mathbb{R})$, $\int_{\mathcal{S}}f(s)\nu^n\pi^{-1}(ds)\to\int_{\mathcal{S}}f(s)\nu\pi^{-1}(ds)$. 
  Observing that for every $n\geq1$, $\nu^n\pi^{-1}=\nu^n_{\mathcal{S}}$ and $\nu\pi^{-1}=\nu_{\mathcal{S}}$ gives us that, for every 
  $f\in\mathcal{C}(\mathcal{S},\mathbb{R})$, $\int_{\mathcal{S}}f(s)\nu^n_{\mathcal{S}}(ds)\to\int_{\mathcal{S}}f(s)\nu_{\mathcal{S}}(ds)$
  . Hence $\lambda(\nu^n)\to\lambda(\nu)$ in $\mathcal{P}(\mathcal{S})$, which gives us continuity of $\lambda(\cdot)$.
  \item [(ii)] Composition of two measurable functions is measurable.
  \item [(iii)] Let $\left\{\gamma_n\right\}_{n\geq1}$ be a sequence in $\mathcal{M}(\mathcal{S}\times U)$ converging to 
  $\gamma\in\mathcal{M}(\mathcal{S}\times U)$. Then, for every $f\in\mathcal{C}(\mathcal{S}\times U,\mathbb{R})$, for every $T>0$ and for 
  every $g\in\mathbb{L}_2[0,T]$, $\int_{0}^{T}g(t)[\int_{\mathcal{S}\times U}f(s,u)\gamma_n(t)(ds,du)]dt\to\int_{0}^{T}g(t)
  [\int_{\mathcal{S}\times U}f(s,u)\gamma(t)(ds,du)]dt$ as $n\to\infty$. Hence, for every $f\in\mathcal{C}(\mathcal{S},\mathbb{R})$, for 
  every $T>0$ and for every $g\in\mathbb{L}_2[0,T]$, $\int_{0}^{T}g(t)[\int_{\mathcal{S}\times U}\left(f\circ\pi\right)(s,u)\gamma_n(t)(ds,du)]dt\to\\
  \int_{0}^{T}g(t)[\int_{\mathcal{S}\times U}\left(f\circ\pi\right)(s,u)\gamma(t)(ds,du)]dt$ as $n\to\infty$. By argument similar to part $(i)$ of this
  lemma, we have, for every $f\in\mathcal{C}(\mathcal{S},\mathbb{R})$, for every $T>0$ and for every $g\in\mathbb{L}_2[0,T]$, 
  $\int_{0}^{T}g(t)[\int_{\mathcal{S}}f(s)\left(\lambda\circ\gamma_n\right)(t)(ds)]dt\to\int_{0}^{T}g(t)
  [\int_{\mathcal{S}}f(s)\left(\lambda\circ\gamma\right)(t)(ds)]dt$ as $n\to\infty$. Therefore $\left\{\left(\lambda\circ\gamma_n\right)\right\}_{n\geq1}$ converges to 
  $\left(\lambda\circ\gamma\right)$ in $\mathcal{M}(\mathcal{S})$, which gives us continuity of $\Lambda(\cdot)$.\qed 
 \end{itemize}
\end{proof}

\section{Recursion analysis}
\label{recanal}

Before we present the analysis of recursion \eqref{rec} we begin with some preliminaries in the next subsection. Later we present the main 
result of the paper followed by the limit set theorem which characterizes the limit set of the recursion in terms of the dynamics of the 
limiting DI. Throughout this section we shall assume that assumptions $(A1)-(A5)$ are satisfied.

\subsection{Preliminaries}\label{prelim}
Define $t(0):=0$ and for every $n\geq1$, $t(n):=\sum_{k=0}^{n-1}a(k)$. Define the stochastic process 
$\bar{X}:\Omega\times [0,\infty)\rightarrow\mathbb{R}^d$ such that, for every $(\omega,t)\in\Omega\times [0,\infty)$, 
\begin{equation*}
 \bar{X}(\omega,t):=\left(\frac{t-t(n)}{t(n+1)-t(n)}\right)X_{n+1}(\omega)+\left(\frac{t(n+1)-t}{t(n+1)-t(n)}\right)X_{n}(\omega),
\end{equation*}
where $n$ is such that $t\in[t(n),t(n+1))$.

Recall that recursion \eqref{rec} is given by,
\begin{equation*}
 X_{n+1}-X_{n}-a(n)M_{n+1}\in H(X_n,S_n),
\end{equation*}
for every $n\geq0$. By Lemma \ref{ctem}, we have that, for every $l\geq1$, for every $n\geq0$, $H(X_n,S_n)\subseteq H^{(l)}(X_n,S_n)$. 
Therefore, for every $l\geq1$, recursion \eqref{rec} can be written as, 
\begin{equation}\label{rec1}
 X_{n+1}-X_{n}-a(n)M_{n+1}\in H^{(l)}(X_n,S_n),
\end{equation}
for every $n\geq0$. By Lemma \ref{param}, we know that for every $l\geq1$, the set valued map $H^{(l)}$, admits a parametrization. The 
next lemma allows us write the recursion in terms of the parametrization of $H^{(l)}$.

\begin{lemma}
 For every $l\geq1$, for every $n\geq0$, there exists a $U$-valued random variable on $\Omega$, say $U^{(l)}_n$, such that,
 \begin{equation*}
  X_{n+1}-X_{n}-a(n)M_{n+1}= h^{(l)}(X_n,S_n,U^{(l)}_n),
 \end{equation*}
 where $h^{(l)}$ is as in Lemma \ref{param}.
\end{lemma}
\begin{proof}Fix $l\geq1$ and $n\geq0$. Define a set valued map $G:\Omega\rightarrow\left\{\text{subsets of }U\right\}$, such that, for every 
 $\omega\in\Omega$, $G(\omega):=\left\{u\in U: h^{(l)}(X_n(\omega),S_n(\omega),u)=X_{n+1}(\omega)-X_n(\omega)-a(n)M_{n+1}(\omega)\right\}$. By 
 Lemma \ref{param} and by $\eqref{rec1}$, we have that for every $\omega\in\Omega$, $G(\omega)\neq\emptyset$. By continuity of 
 $h^{(l)}(\cdot)$, we have that for every $\omega\in\Omega$, $G(\omega)$ is closed. For any $C\subseteq U$, closed, 
 $G^{-1}(C)=\left\{\omega\in\Omega:d(X_{n+1}(\omega)-X_n(\omega)-a(n)M_{n+1}(\omega),h^{(l)}(X_n(\omega),S_n(\omega),C))=0\right\}$. Clearly the set 
 valued map $\omega\rightarrow h^{(l)}(X_n(\omega),S_n(\omega),C)$ is measurable and hence by Lemma \ref{mtool}, the map
 $\omega\rightarrow d(X_{n+1}(\omega)-X_n(\omega)-a(n)M_{n+1}(\omega),h^{(l)}(X_n(\omega),S_n(\omega),C))$ is measurable. Thus 
 $G^{-1}(C)\in\mathscr{F}$. Therefore by Lemma \ref{msel}, $\mathscr{S}(G)\neq\emptyset$. Set $U^{(l)}_n$ to be any measurable selection 
 of $G$.\qed
\end{proof}

For every $l\geq1$, define $\Gamma^{(l)}:\Omega\times[0,\infty)\rightarrow\mathcal{P}(\mathcal{S}\times U)$, such that, for every 
$(\omega,t)\in\Omega\times[0,\infty)$,
\begin{equation}\label{diracm}
 \Gamma^{(l)}(\omega,t)=\delta_{S_n(\omega)}\otimes\delta_{U^{(l)}_n(\omega)},
\end{equation}
where $n$ is such that $t\in[t(n),t(n+1))$, $\delta_{S_n(\omega)}\in\mathcal{P}(\mathcal{S})$ denotes the Dirac measure in 
$\mathcal{P}(\mathcal{S})$(i.e. for any $A\in\mathscr{B}(\mathcal{S})$, $\delta_{S_n(\omega)}(A)=1$ if $S_n(\omega)\in A$ or $0$ otherwise) 
and $\delta_{U^{(l)}_n(\omega)}\in\mathcal{P}(U)$ denotes the Dirac measure in $\mathcal{P}(U)$.

For every $l\geq1$, for every $\nu\in\mathcal{P}(\mathcal{S}\times U)$, define $\tilde{h}^{(l)}_{\nu}:\mathbb{R}^d\rightarrow\mathbb{R}^d$, 
such that for every $x\in\mathbb{R}^d$, $\tilde{h}^{(l)}_{\nu}(x):=\int_{\mathcal{S}\times U}h^{(l)}(x,s,u)\nu(ds,du)$ where 
$h^{(l)}(\cdot)$ are as in Lemma \ref{param}. The next lemma provides an equicontinuity result used later.

\begin{lemma}\label{paramuse}
 For every $l\geq1$, for every $r>0$, $\left\{\tilde{h}^{(l)}_{\nu}(\cdot)|_{rU}:\nu\in\mathcal{P}(\mathcal{S}\times U)\right\}$ is an 
 equicontinuous family where $\tilde{h}^{(l)}_{\nu}(\cdot)|_{rU}$ denotes the restriction of $\tilde{h}^{(l)}_\nu(\cdot)$ to the closed 
 ball of radius $r$.
\end{lemma}
\begin{proof}
 Fix $l\geq1$ and $r>0$. Since $rU\times \mathcal{S}\times U$ is compact, by Heine-Cantor theorem, 
 $h^{(l)}(\cdot)|_{rU\times \mathcal{S}\times U}$ is uniformly continuous. Fix $\epsilon>0$. Then, there exits $\delta>0$
 (depending on $l$, $r$ and $\epsilon$), such that, for every $x,x'\in rU$, for every $s,s'\in\mathcal{S}$, for every $u,u'\in U$ 
 satisfying, $\parallel x-x'\parallel<\delta$, $d_{\mathcal{S}}(s,s')<\delta$, $\parallel u-u'\parallel<\delta$, we have, 
 $\parallel h^{(l)}(x,s,u)-h^{(l)}(x',s',u')\parallel<\epsilon$. Therefore for $\delta$ as above, for $x,x'\in rU$ satisfying 
 $\parallel x-x'\parallel<\delta$, we have that for every $\nu\in\mathcal{P}(\mathcal{S}\times U)$, $\parallel \tilde{h}^{(l)}_{\nu}(x)-
 \tilde{h}^{(l)}_{\nu}(x')\parallel=\parallel \int_{\mathcal{S}\times U}(h^{(l)}(x,s,u)-h^{(l)}(x',s,u))\nu(ds,du)\parallel\leq
 \\\int_{\mathcal{S}\times U}\parallel h^{(l)}(x,s,u)-h^{(l)}(x',s,u)\parallel\nu(ds,du)<\int_{\mathcal{S}\times U}\epsilon\nu(ds,du)=
 \epsilon$.\qed
\end{proof}

For every $l\geq1$, define the stochastic process, $Y^{(l)}:\Omega\times[0,\infty)\rightarrow\mathbb{R}^d$, such that, for every 
$(\omega,t)\in\Omega\times[0,\infty)$, $Y^{(l)}(\omega,t):=h^{(l)}(X_n(\omega),S_n(\omega),U^{(l)}_n(\omega))$, where $n$ is such that 
$t\in[t(n),t(n+1))$. 

In what follows, in the next two subsections, most of the arguments are sample path wise. We omit $\omega$ from our notation and use lower 
case letters to denote the above defined quantities along a particular sample path, for example, $\bar{x}(t),\ x_n,\ s_n,\ m_n,\ 
u^{(l)}_n,\ y^{(l)}(t)$ and $\gamma^{(l)}(t)$ are to be understood as $\bar{X}(\omega,t),\ X_n(\omega),\ S_n(\omega),\ M_n(\omega),\ 
U^{(l)}_n(\omega),\ Y^{(l)}(\omega,t)$ and $\Gamma^{(l)}(\omega,t)$ respectively for some $\omega$ fixed. 

\subsection{Main result - Asymptotic pseudotrajectory}

For every $\omega\in\Omega$, for every $l\geq1$, for every $\tilde{t}\geq0$, let $\tilde{x}^{(l)}(\cdot;\tilde{t})$ denote the solution 
of the o.d.e., 
\begin{equation}
\label{ode}
 \dot{\tilde{x}}^{(l)}(t;\tilde{t})=y^{(l)}(t+\tilde{t}),
\end{equation}
for every $t\geq0$, with initial condition $\tilde{x}^{(l)}(0;\tilde{t})=\bar{x}(\tilde{t})$. 

First we shall get rid of the additive noise terms, $M_n$. Let 
\begin{equation*}
 \Omega_{a,s}:=\left\{\omega\in\Omega: (A4)\ and\ (A5)\ hold\right\},
\end{equation*}
i.e., the set of sample paths where the iterates are stable and the additive noise terms are eventually negligible. By assumption $(A4)$ 
and $(A5)$, we have $\mathbb{P}(\Omega_{a,s})=1$.

\begin{lemma}\label{addnois}
 For every $l\geq1$, almost surely for every $\omega$, for every $T>0$,
 \begin{equation*}
  \lim_{t\to\infty}\sup_{0\leq q\leq T}\parallel \bar{x}(t+q)-\tilde{x}^{(l)}(q;t)\parallel=0.
 \end{equation*}
\end{lemma}
\begin{proof} 
Fix $l\geq1$, $\omega\in\Omega_{a,s}$ and $T>0$. We shall prove the claim along the sequence $\left\{t(n)\right\}_{n\geq0}$ as defined in section 
\ref{prelim}. The general claim easily follows from this special case.

Fix $n\geq0$. Let $\tau(n,T):=\min\left\{k>n:t(k)\geq t(n)+T\right\}$. Let $q\in[0,T]$. Then, there exists $k$ such that $t(n)+q\in[t(k),t(k+1))$ 
and $n\leq k\leq \tau(n,T)-1$. By definition of $\bar{x}(\cdot)$ and $\tilde{x}^{(l)}(\cdot;t(n))$, we have that, 
$\bar{x}(t(n)+q)=\alpha x_k+(1-\alpha)x_{k+1}$ and $\tilde{x}^{(l)}(q;t(n))=\alpha\tilde{x}^{(l)}(t(k)-t(n);t(n))+\\(1-\alpha)
\tilde{x}^{(l)}(t(k+1)-t(n);t(n))$ where $\alpha=\frac{t(k+1)-t(n)-q}{t(k+1)-t(k)}$. Since $\tilde{x}^{(l)}(\cdot;t(n))$ is a solution of the 
o.d.e. $\eqref{ode}$, we have that, for every $k\geq n$, $\tilde{x}^{(l)}(t(k)-t(n);t(n))=x_n+\sum_{j=n}^{k-1}a(j)h^{(l)}(x_j,s_j,u^{(l)}_j)
$ and by Lemma \ref{paramuse}, we have that, for every $k\geq n$, $x_k=\bar{x}(t(k))=x_n+\sum_{j=n}^{k-1}a(j)h^{(l)}(x_j,s_j,u^{(l)}_j)+
\sum_{j=n}^{k-1}a(j)m_{j+1}$. Thus, 
\begin{align*}
\parallel\bar{x}(t(n)+q)-\tilde{x}^{(l)}(q;t(n))\parallel&\leq \parallel\alpha\sum_{j=n}^{k-1}a(j)m_{j+1}+(1-\alpha)
                              \sum_{j=n}^{k}a(j)m_{j+1}\parallel\\
                              &\leq\alpha\parallel\sum_{j=n}^{k-1}a(j)m_{j+1}\parallel+(1-\alpha)\parallel\sum_{j=n}^{k}a(j)m_{j+1}
                              \parallel\\
                              &\leq \sup_{n\leq k\leq \tau(n,T)}\parallel\sum_{j=n}^{k}a(j)m_{j+1}\parallel.
\end{align*}
Since the r.h.s. of the above inequality is independent of $q\in[0,T]$, we have, $\\\sup_{0\leq q\leq T}\parallel\bar{x}(t(n)+q)-
\tilde{x}^{(l)}(q;t(n))\parallel\leq\sup_{n\leq k\leq \tau(n,T)}\parallel\sum_{j=n}^{k}a(j)m_{j+1}\parallel$. Therefore, 
$\\\lim_{n\to\infty}\sup_{0\leq q\leq T}\parallel\bar{x}(t(n)+q)-\tilde{x}^{(l)}(q;t(n))\parallel\leq\lim_{n\to\infty}
\sup_{n\leq k\leq \tau(n,T)}\parallel\sum_{j=n}^{k}a(j)m_{j+1}\parallel$. Now the claim follows follows from assumption $(A4)$.\qed
\end{proof}
\begin{lemma}\label{eqcont}
 For every $l\geq1$, for almost every $\omega$, $\left\{\tilde{x}^{(l)}(\cdot;t)\right\}_{t\geq0}$ is relatively compact in 
 $\mathcal{C}([0,\infty),\mathbb{R}^d)$.
\end{lemma}
\begin{proof}
 Fix $l\geq1$, $\omega\in\Omega_{a,s}$. By assumption $(A5)$, we know that there exists $r>0$ such that 
 $\sup_{n\geq0}\parallel x_n\parallel\leq r$ and hence $\sup_{t\geq0}\tilde{x}^{(l)}(0;t)=\sup_{t\geq0}\bar{x}(t)\leq r$. 
 
 For any $t\geq0$, let $[t]:=\max\left\{n\geq0:t(n)\leq t\right\}$. For every $t\geq0$ and $q_1,q_2\in[0,\infty)$ (w.l.o.g. assume $q_1<q_2$) 
 we have, 
 \begin{align*}
  \parallel \tilde{x}^{(l)}(q_1;t)-\tilde{x}^{(l)}(q_2;t)\parallel&=\parallel\int_{q_1}^{q_2}h^{(l)}(x_{[t+q]},s_{[t+q]},u^{(l)}_{[t+q]})
                                                                     dq\parallel\\
                                         &\leq\int_{q_1}^{q_2}\parallel h^{(l)}(x_{[t+q]},s_{[t+q]},u^{(l)}_{[t+q]})\parallel dq\\
                                         &\leq\int_{q_1}^{q_2}K^{(l)}(1+\parallel x_{[t+q]}\parallel)dq\\
                                         &\leq C^{(l)}(q_2-q_1),
 \end{align*}
where $C^{(l)}:=K^{(l)}(1+r)$ and $r>0$ is such that, $\sup_{n\geq0}\parallel x_n\parallel\leq r$. Thus 
$\left\{\tilde{x}^{(l)}(\cdot;t)\right\}_{t\geq0}$ is an equicontinuous family. Now the claim follows from Arzella-Ascoli theorem.
\qed  
\end{proof}

From Lemma \ref{addnois} and Lemma \ref{eqcont} we conclude the following for almost every $\omega$.
\begin{itemize}
 \item [(1)] The family of functions, \bf{$\left\{\bar{x}(\cdot+t)\right\}_{t\geq0}$ is relatively compact }\rm in 
 $\mathcal{C}([0,\infty),\mathbb{R}^d)$, because if not, there exist $t_n\to\infty$ such that $\left\{\bar{x}(\cdot+t_n)\right\}_{n\geq0}$ does not 
 have a limit point in $\mathcal{C}([0,\infty),\mathbb{R}^d)$. Then by Lemma \ref{addnois}, $\left\{\tilde{x}^{(l)}(\cdot;t_n)\right\}_{n\geq0}$ does 
 not have a limit point in $\mathcal{C}([0,\infty),\mathbb{R}^d)$ which contradicts Lemma \ref{eqcont}.
 \item [(2)] Since $\left\{\bar{x}(\cdot+t)\right\}_{t\geq0}$ is relatively compact in $\mathcal{C}([0,\infty),\mathbb{R}^d)$, by Arzella-Ascoli 
 theorem, we have that for every $T>0$, $\left\{\bar{x}(\cdot+t)|_{[0,T]}\right\}_{t\geq0}$ is equicontinuous. Set $T=1$ and fix $\epsilon>0$. Then 
 for any $t_0\geq\frac{1}{2}$, there exists $t\geq0$ such that $t+\frac{1}{2}=t_0$. By equicontinuity of $\left\{\bar{x}(\cdot+t)|_{[0,T]}\right\}_{t\geq0}$, 
 we can obtain a $\delta>0$ (independent of $t$ and hence $t_0$) such that, for every $t'$ satisfying 
 $|t'-t_0|<\min\left\{\delta,\frac{1}{2}\right\}$, we have $\parallel \bar{x}(t')-\bar{x}(t_0)\parallel<\epsilon$. Since $\epsilon$ was arbitrary 
 and $\bar{x}(\cdot)|_{[0,\frac{1}{2}]}$ is uniformly continuous we obtain that the function \bf{$\bar{x}(\cdot)$ is uniformly continuous
 }\rm on $[0,\infty)$.
\end{itemize}
 
\begin{proposition}\label{main1}
 For $a.e.\ \omega$, every limit point $x^*(\cdot)$ of $\left\{\bar{x}(\cdot+t)\right\}_{t\geq0}$ satisfies the following.
 \begin{itemize}
  \item [(i)] For every $l\geq1$, there exists $\tilde{\gamma}^{(l)}\in\mathcal{M}(\mathcal{S}\times U)$ such that, for every $t\geq0$,
  \begin{equation*}
   x^*(t)=x^*(0)+\int_{0}^{t}\left[\int_{\mathcal{S}\times U}h^{(l)}(x^*(q),s,u)\tilde{\gamma}^{(l)}(q)(ds,du)\right]dq.
  \end{equation*}
  \item [(ii)] For every $l\geq1$, $\tilde{\gamma}^{(l)}$ as in part $(i)$ of this lemma is such that, for almost every $t\geq0$,
  \begin{equation*}
  \Lambda(\tilde{\gamma}^{(l)})(t)\in D(x^*(t)).
  \end{equation*}
  \item [(iii)] $x^*(\cdot)$ is absolutely continuous and for almost every $t\geq0$,
  \begin{equation*}
   \frac{dx^*(t)}{dt}\in\hat{H}(x^*(t)).
  \end{equation*}
 \end{itemize} 
\end{proposition}
\begin{proof}
 Fix $\omega\in\Omega_{a,s}$, and let $t_n\to\infty$, such that $\bar{x}(\cdot+t_n)\to x^*(\cdot)$ in 
 $\mathcal{C}([0,\infty),\mathbb{R}^d)$.
 \begin{itemize}
  \item [(i)] Fix $l\geq1$. Consider the $\mathcal{M}(\mathcal{S}\times U)$ valued sequence $\left\{\gamma^{(l)}(\cdot+t_n)\right\}_{n\geq1}$. Since 
  $\mathcal{M}(\mathcal{S}\times U)$ is a compact metric space, there exists a subsequence of the above that converges. Set 
  $\tilde{\gamma}^{(l)}(\cdot)$ to be some limit point and w.l.o.g. assume $\left\{\gamma^{(l)}(\cdot+t_n)\right\}_{n\geq1}$ converges to 
  $\tilde{\gamma}^{(l)}(\cdot)$. Since the sequence $\left\{\bar{x}(\cdot+t_n)\right\}_{n\geq1}$ converges to $x^*(\cdot)$, by Lemma \ref{addnois}, 
  we have that $\left\{\tilde{x}^{(l)}(\cdot;t_n)\right\}_{n\geq1}$ also converges to $x^*(\cdot)$ in $\mathcal{C}([0,\infty),\mathbb{R}^d)$. For every 
  $n\geq1$, by definition of $\tilde{x}^{(l)}(\cdot;t_n)$ we have that for every $t\geq0$, 
  \begin{align*}
  \tilde{x}^{(l)}(t;t_n)&=\bar{x}(t_n)+\int_{0}^{t}y^{(l)}(q)dq\\
                  &=\bar{x}(t_n)+\int_{0}^{t}h^{(l)}(x_{[t_n+q]},s_{[t_n+q]},u^{(l)}_{[t_n+q]})dq. 
  \end{align*}
  By definition of $\gamma^{(l)}(\cdot)$(see \eqref{diracm} and recall that $\gamma^{(l)}(\cdot)=\Gamma^{(l)}(\omega,\cdot)$) for every 
  $n\geq1$ and for every $t\geq0$ we can write the above as,
  \begin{equation*}
   \tilde{x}^{(l)}(t;t_n)=\bar{x}(t_n)+\int_{0}^{t}\left[\int_{\mathcal{S}\times U}h^{(l)}(x_{[t_n+q]},s,u)\gamma^{(l)}(q+t_n)(ds,du)\right]dq.
  \end{equation*}
  Therefore, for every $t\geq0$,
  \begin{align}\label{tmp0}
   \lim_{n\to\infty}[\tilde{x}^{(l)}(t;t_n)-\bar{x}(t_n)]&=\lim_{n\to\infty}\int_{0}^{t}\left[\int_{\mathcal{S}\times U}
                                                   h^{(l)}(x_{[t_n+q]},s,u)\gamma^{(l)}(q+t_n)(ds,du)\right]dq.\nonumber\\
                                      x^*(t)-x^*(0)&=\lim_{n\to\infty}\int_{0}^{t}\left[\int_{\mathcal{S}\times U}
                                                   h^{(l)}(x_{[t_n+q]},s,u)\gamma^{(l)}(q+t_n)(ds,du)\right]dq.
  \end{align}
  Since $\gamma^{(l)}(\cdot+t_n)\to\tilde{\gamma}^{(l)}(\cdot)$ and by our choice of the topology for $\mathcal{M}(\mathcal{S}\times U)$, 
  we have, 
  \begin{equation*}
   \int_{0}^{t}\left[\int_{\mathcal{S}\times U}\tilde{f}(q,s,u)\gamma^{(l)}(q+t_n)(ds,du)\right]dq-
   \int_{0}^{t}\left[\int_{\mathcal{S}\times U}\tilde{f}(q,s,u)\tilde{\gamma}^{(l)}(q)(ds,du)\right]dq\to0,
  \end{equation*}
 for all  bounded continuous $\tilde{f}:[0,t]\times\mathcal{S}\times U\rightarrow\mathbb{R}$ of the form, 
 \begin{equation*}
  \tilde{f}(q,s,u)=\sum_{m=1}^{N}a_mg_m(q)f_m(s,u),
 \end{equation*}
 for some $N\geq1$, scalars $a_m$ and bounded continuous functions $g_m,\ f_m$ on $[0,t],\ \mathcal{S}\times U$ respectively, 
 for $1\leq m\leq N$. By the Stone-Weierstrass theorem, such functions can uniformly approximate any function in 
 $\mathcal{C}([0,t]\times\mathcal{S}\times U,\mathbb{R})$. Thus the above convergence holds true for all real valued continuous functions 
 on $[0,t]\times\mathcal{S}\times U$, implying that $\\t^{-1}\gamma^{(l)}(q+t_n)(ds,du)dq\to t^{-1}\tilde{\gamma}^{(l)}(q)(ds,du)dq$ in 
 $\mathcal{P}([0,t]\times\mathcal{S}\times U)$. Thus,
 \begin{small}
 \begin{equation}\label{tmp1}
  \parallel\int_{0}^{t}\left[\int_{\mathcal{S}\times U}h^{(l)}(x^*(q),s,u)\gamma^{(l)}(q+t_n)(ds,du)\right]dq-
  \int_{0}^{t}\left[\int_{\mathcal{S}\times U}h^{(l)}(x^*(q),s,u)\tilde{\gamma}^{(l)}(q)(ds,du)\right]dq\parallel\to0
 \end{equation}
 \end{small}
 as $n\to\infty$. Since $\left\{\bar{x}(\cdot+t_n)|_{[0,t]}\right\}_{n\geq1}$ converges uniformly to $x^*(\cdot)|_{[0,t]}$ we have that, the 
 function $q\rightarrow x_{[t_n+q]}$ converges uniformly to $x^*(\cdot)|_{[0,t]}$ on $[0,t]$. Using the above and by 
 Lemma \ref{paramuse}, we have that for every $\epsilon>0$, there exists $N$ (depending on $\epsilon$) such that, for every $n\geq N$, for 
 every $q\in[0,t]$, we have, 
 \begin{equation}\label{tmp2}
  \parallel\int_{\mathcal{S}\times U}h^{(l)}(x_{[t_n+q]},s,u)\gamma^{(l)}(q+t_n)(ds,du)-
  \int_{\mathcal{S}\times U}h^{(l)}(x^*(q),s,u)\gamma^{(l)}(q+t_n)(ds,du)\parallel<\epsilon.
 \end{equation}
 Now,
 \begin{small}
  \begin{align*}
   \parallel\int_{0}^{t}\left[\int_{\mathcal{S}\times U}h^{(l)}(x_{[t_n+q]},s,u)\gamma^{(l)}(q+t_n)(ds,du)\right]dq&-
  \int_{0}^{t}\left[\int_{\mathcal{S}\times U}h^{(l)}(x^*(q),s,u)\tilde{\gamma}^{(l)}(q)(ds,du)\right]dq\parallel\\
  &\leq\\
  \parallel\int_{0}^{t}\int_{\mathcal{S}\times U}h^{(l)}(x_{[t_n+q]},s,u)\gamma^{(l)}(q+t_n)(ds,du)dq&-
  \int_{0}^{t}\int_{\mathcal{S}\times U}h^{(l)}(x^*(q),s,u)\gamma^{(l)}(q+t_n)(ds,du)dq\parallel\\
  &+\\
  \parallel\int_{0}^{t}\left[\int_{\mathcal{S}\times U}h^{(l)}(x^*(q),s,u)\gamma^{(l)}(q+t_n)(ds,du)\right]dq&-
  \int_{0}^{t}\left[\int_{\mathcal{S}\times U}h^{(l)}(x^*(q),s,u)\tilde{\gamma}^{(l)}(q)(ds,du)\right]dq\parallel.
  \end{align*}
 \end{small}
 Taking limit on both sides as $n\to\infty$ in the above equation and using equations \eqref{tmp1} and \eqref{tmp2} we obtain, 
 \begin{small}
 \begin{equation*}
  \lim_{n\to\infty}\parallel\int_{0}^{t}\left[\int_{\mathcal{S}\times U}h^{(l)}(x_{[t_n+q]},s,u)\gamma^{(l)}(q+t_n)(ds,du)\right]dq-
  \int_{0}^{t}\left[\int_{\mathcal{S}\times U}h^{(l)}(x^*(q),s,u)\tilde{\gamma}^{(l)}(q)(ds,du)\right]dq\parallel\leq\epsilon t,
 \end{equation*}
 \end{small}
 for every $\epsilon>0$. Therefore, for every $t\geq0$,
 \begin{small}
  \begin{equation*}
   \lim_{n\to\infty}\int_{0}^{t}\left[\int_{\mathcal{S}\times U}h^{(l)}(x_{[t_n+q]},s,u)\gamma^{(l)}(q+t_n)(ds,du)\right]dq=
   \int_{0}^{t}\left[\int_{\mathcal{S}\times U}h^{(l)}(x^*(q),s,u)\tilde{\gamma}^{(l)}(q)(ds,du)\right]dq.
  \end{equation*}
 \end{small}
 Substituting the above limit in equation \eqref{tmp0}, we get, for every $t\geq0$,
 \begin{equation*}
  x^*(t)-x^*(0)=\int_{0}^{t}\left[\int_{\mathcal{S}\times U}h^{(l)}(x^*(q),s,u)\tilde{\gamma}^{(l)}(q)(ds,du)\right]dq.
 \end{equation*}
 \item [(ii)] The proof of this part is similar to the proof of Lemma 6, chapter 6.3 of \cite{borkartxt}. We shall present a proof here for 
 the sake of completeness. 
 Let $\left\{f_i\right\}$ be a countable set of real valued continuous functions on $\mathcal{S}$ that is a convergence determining class for 
 $\mathcal{P}(\mathcal{S})$. By replacing each $f_i$ by $a_if_i+b_i$ for suitable scalars $a_i,\ b_i>0$, we may suppose that 
 $0\leq f_i(\cdot)\leq 1$ for all $i$. For each $i$,
 \begin{equation*}
  \zeta_n^{i}:=\sum_{k=0}^{n-1}a(k)(f_i(S_{k+1})-\int_{\mathcal{S}}f_i(s')\Pi(X_k,S_k)(ds')),
 \end{equation*}
 is a square integrable zero mean martingale w.r.t. the filtration $\left\{\mathscr{F}_n:=\sigma(S_k,X_k:0\leq k\leq n-1)\right\}_{n\geq1}$ and for 
 almost every $\omega$, $\sum_{n=1}^{\infty}\mathbb{E}[(\zeta^i_{n+1}-\zeta^i_{n})^2|\mathscr{F}_n]\leq 2\sum_{n=0}^{\infty}a(n)^2<\infty$
 . Hence by martingale convergence theorem (see Appendix C, Theorem 11 in \cite{borkartxt}), for almost every $\omega$, 
 $\left\{\zeta_n^i\right\}_{n\geq1}$ converges. Let $\Omega_{m}:=\left\{\omega\in\Omega: \forall i, \left\{\zeta_n^i\right\}_{n\geq1}\ converges\right\}$. Then 
 $\mathbb{P}(\Omega_{m})=1$. Define,
 \begin{equation}
  \Omega^*:=\Omega_m\cap\Omega_{a,s}
 \end{equation}
 and clearly $\mathbb{P}(\Omega^*)=1$. Recall, that for every $T\geq0$ and $n\geq0$, $\tau(n,T):=\min\left\{k\geq n:t(k)\geq t(n)+T\right\}$. Then 
 for every $\omega\in\Omega^*$, for every $i$, for every $T>0$, as $n\to\infty$,
 \begin{equation*}
  \sum_{k=n}^{\tau(n,T)}a(k)(f_i(s_{k+1})-\int_{\mathcal{S}}f_i(s')\Pi(x_k,s_k)(ds'))\to0.
 \end{equation*}
 By the choice of $\left\{f_i\right\}_{i\geq1}$ and the fact that $\left\{a(n)\right\}_{n\geq0}$ are non-increasing (see assumption $(A3)(i)$) we get, that for 
 every $\omega\in\Omega^*$, for every $i$, for every $T>0$, as $n\to\infty$,
 \begin{equation*}
  |\sum_{k=n}^{\tau(n,T)}(a(k)-a(k+1))f_i(s_{k+1})|\leq a(n)-a(\tau(n,T)+1)\to0.
 \end{equation*}
Thus, for every $\omega\in\Omega^*$, for every $i$, for every $T>0$, as $n\to\infty$,
\begin{equation}\label{tmp3}
 \sum_{k=n}^{\tau(n,T)}a(k)(f_i(s_k)-\int_{\mathcal{S}}f_i(s')\Pi(x_k,s_k)(ds'))\to0.
\end{equation}
Fix $l\geq1$. Define for every $\omega$, $\mu^{(l)}(\cdot):=\Lambda(\gamma^{(l)})\in\mathcal{M}(\mathcal{S})$. Then by definition of 
$\gamma^{(l)}(\cdot)$ (see equation \eqref{diracm} and recall $\gamma^{(l)}(\cdot):=\Gamma^{(l)}(\omega,\cdot)$) and $\Lambda(\cdot)$ (see 
Lemma \ref{lambd}) we have that for every $\omega$, for every $t>0$, $\mu^{(l)}(t)=\delta_{s_n}$ where $n$ is such that, 
$t\in[t(n),t(n+1))$. Using the definition of $\mu^{(l)}(\cdot)$ in equation \eqref{tmp3}, we get that for every $\omega\in\Omega^*$, for 
every $i$, for every $T>0$, as $n\to\infty$,
\begin{equation*}
 \int_{0}^{t(\tau(n,T))-t(n)}\int_{\mathcal{S}}\left[f_i(s)-\int_{\mathcal{S}}f_i(s')\Pi(x_{[t(n)+q]},s)(ds')\right]\mu^{(l)}(q+t(n))
 (ds)dq\to0,
\end{equation*}
where, for every $t\geq0$, $[t]:=\max\left\{n\geq0:t(n)\leq t\right\}$. From the above it can easily be shown that for every $\omega\in\Omega^*$, 
for every $i$, for every $T>0$,
\begin{equation*}
 \lim_{t\to\infty}\int_{0}^{T}\int_{\mathcal{S}}\left[f_i(s)-\int_{\mathcal{S}}f_i(s')\Pi(x_{[t+q]},s)(ds')\right]\mu^{(l)}(q+t)(ds)dq=0.
\end{equation*}
By assumption $(A2)$, $\Pi(\cdot)|_{2rU\times\mathcal{S}}$ is uniformly continuous, where $r=\sup_{n\geq0}\parallel x_n\parallel$. 
Thus the function $(x,s)\rightarrow f_i(s)-\int_{\mathcal{S}}f_i(s')\Pi(x,s)(ds')$ is uniformly continuous on $2rU\times\mathcal{S}$ for 
every $i$. Using the above, the fact that $\lim_{n\to\infty}(t(n+1)-t(n))=0$ and uniform continuity of $\bar{x}(\cdot)$ we get that, for 
every $\omega\in\Omega^*$, for every $i$, for every $T>0$,
\begin{equation*}
\lim_{t\to\infty}\int_{0}^{T}\int_{\mathcal{S}}\left[f_i(s)-\int_{\mathcal{S}}f_i(s')\Pi(\bar{x}(t+q),s)(ds')\right]\mu^{(l)}(q+t)(ds)
dq=0.
\end{equation*}
Fix $\omega\in\Omega^*$. From part $(i)$ of this lemma we have a sequence $t_n\to\infty$ such that $\left\{\bar{x}(\cdot+t_n)\right\}_{n\geq1}$ 
converges to $x^*(\cdot)$ in $\mathcal{C}([0,\infty),\mathbb{R}^d)$ and $\left\{\gamma^{(l)}(\cdot+t_n)\right\}_{n\geq1}$ converges to 
$\tilde{\gamma}^{(l)}(\cdot)$ in $\mathcal{M}(\mathcal{S}\times U)$. By continuity of the map $\Lambda(\cdot)$ (see Lemma \ref{lambd}), 
we have that $\mu^{(l)}(\cdot+t_n)\to\tilde{\mu}^{(l)}(\cdot)=\Lambda(\tilde{\gamma}^{(l)})$ in $\mathcal{M}(\mathcal{S})$. Using the 
convergence above and the fact that the family of functions, $\left\{x\in2rU\rightarrow\int_{\mathcal{S}}\left[f_i(s)-\int_{\mathcal{S}}
f_i(s')\Pi(x,s)(ds')\right]\nu(ds):\nu\in\mathcal{P}(\mathcal{S})\right\}$ is equicontinuous (which can be shown by arguments similar to Lemma 
\ref{paramuse}) we get that for every $i$, for every $T>0$,
\begin{equation*}
 \int_{0}^{T}\int_{\mathcal{S}}\left[f_i(s)-\int_{\mathcal{S}}f_i(s')\Pi(x^*(q),s)(ds')\right]\tilde{\mu}^{(l)}(q)(ds)dq=0.
\end{equation*}
An application of Lesbesgue's theorem (see chapter 11.1.3 in \cite{borkartxt}), we get that for almost every $t\geq0$, for every $i$, 
\begin{equation*}
 \int_{\mathcal{S}}\left[f_i(s)-\int_{\mathcal{S}}f_i(s')\Pi(x^*(t),s)(ds')\right]\tilde{\mu}^{(l)}(t)(ds)=0.
\end{equation*}
By our choice of $\left\{f_i\right\}_{i\geq1}$, we get that for almost every $t\geq0$,
\begin{equation*}
 \tilde{\mu}^{(l)}(t)(ds)=\int_{\mathcal{S}}\Pi(x^*(t),s')(ds)\tilde{\mu}^{(l)}(t)(ds').
\end{equation*}
Therefore for almost every $t\geq0$, $\tilde{\mu}^{(l)}(t)=\Lambda(\tilde{\gamma}^{(l)})(t)\in D(x^*(t))$.
\item [(iii)] Fix $l\geq1$. From part $(i)$ of this lemma we have that for every $t\geq0$,
\begin{equation*}
 x^*(t)=x^*(0)+\int_{0}^{t}\left[\int_{\mathcal{S}\times U}h^{(l)}(x^*(q),s,u)\tilde{\gamma}^{(l)}(q)(ds,du)\right]dq.
\end{equation*}
Clearly $x^*(\cdot)$ is absolutely continuous and for almost every $t\geq0$,
\begin{equation}\label{tmp4}
 \frac{dx^*(t)}{dt}=\int_{\mathcal{S}\times U}h^{(l)}(x^*(t),s,u)\tilde{\gamma}^{(l)}(t)(ds,du).
\end{equation}
By part $(ii)$ of this lemma we have that for almost every $t\geq0$, $\Lambda(\tilde{\gamma}^{(l)})(t)\in D(x^*(t))$. By definition of 
the map $\Lambda(\cdot)$ (see Lemma \ref{lambd}), we have that for almost every $t\geq0$, $\tilde{\gamma}^{(l)}_{\mathcal{S}}(t)\in 
D(x^*(t))$. By Lemma \ref{chint} and by definition of $\hat{H}^{(l)}$ (see equation \eqref{lim2}) we have that for almost every $t\geq0$,
\begin{equation*}
 \int_{\mathcal{S}\times U}h^{(l)}(x^*(t),s,u)\tilde{\gamma}^{(l)}(t)(ds,du)\in
 \cup_{\mu\in D(x^*(t))}\int_{\mathcal{S}}H^{(l)}_{x^*(t)}(s)\mu(ds)=\hat{H}^{(l)}(x^*(t)).
\end{equation*}
Using the above in equation \eqref{tmp4} we obtain that for almost every $t\geq0$,
\begin{equation*}
 \frac{dx^*(t)}{dt}\in\hat{H}^{(l)}(x^*(t)).
\end{equation*}
Since $l\geq1$ that was fixed was arbitrary, the above holds for every $l\geq1$. Therefore, for almost every $t\geq0$,
\begin{equation*}
 \frac{dx^*(t)}{dt}\in\cap_{l\geq1}\hat{H}^{(l)}(x^*(t))=\hat{H}(x^*(t)),
\end{equation*}
where the equality follows from Lemma \ref{approx1}$(iii)$.\qed
\end{itemize}
\end{proof}

Before we proceed further we shall briefly recall the definition of asymptotic pseudotrajectories (APT) for set-valued dynamics 
introduced in \cite{benaim1}. The translation flow $\Theta:\mathcal{C}(\mathbb{R},\mathbb{R}^d)\times \mathbb{R}\rightarrow
\mathcal{C}(\mathbb{R},\mathbb{R}^d)$ is the flow defined by,
\begin{equation*}
 \Theta^t(\bf{x}\rm)(q):=\bf{x}\rm(q+t).
\end{equation*}
For every $\omega\in\Omega$, extend $\bar{x}(\cdot)$ to $\mathbb{R}$ by letting $\bar{x}(t)=\bar{x}(0)$ for $t<0$. Then $\bar{x}(\cdot)$ 
is an APT for the flow of DI \eqref{di2} if, 
\begin{equation*}
 \lim_{t\to\infty}\bf{D}\rm(\Theta^t(\bar{x}),\Sigma)=0,
\end{equation*}
where, $\Sigma:=\cup_{x\in\mathbb{R}^d}\Sigma(x)$ denotes the set of all solutions of DI \eqref{di2}.

In what follows we fix $\omega\in\Omega^*$ and let $\bar{x}(\cdot)$ denote the extension to $\mathbb{R}$ as defined above.
By uniform continuity of $\bar{x}(\cdot)$, we have that the family $\left\{\Theta^t(\bar{x})\right\}_{t\geq0}$ is equicontinuous and by 
assumption $(A5)$ is pointwise bounded. Hence $\left\{\Theta^{t}(\bar{x})\right\}_{t\geq0}$ is relatively compact in 
$\mathcal{C}(\mathbb{R},\mathbb{R}^d)$. Let $x^*(\cdot)$ be a limit point of $\left\{\Theta^{t}(\bar{x})\right\}_{t\geq0}$. Then by Proposition 
\ref{main1}$(iii)$, we have that $x^*(\cdot)|_{[0,\infty)}$ is a solution on $[0,\infty)$ of DI \eqref{di2}. Usually the negative time 
argument is omitted since it follows from the positive time argument as follows:

Fix $T>0$. Since $x^*(\cdot)$ is a limit point of $\left\{\Theta^{t}(\bar{x})\right\}_{t\geq0}$, there exists $t_n\to\infty$ such that, 
$\left\{\Theta^{t_n}(\bar{x})\right\}_{n\geq1}$ converges to $x^*(\cdot)$ in $\mathcal{C}(\mathbb{R},\mathbb{R}^d)$. Then 
$\left\{\Theta^{t_n-T}(\bar{x})\right\}_{n\geq1}$ converges to $x^*(\cdot-T)$. By Proposition \ref{main1}$(iii)$, $x^*(\cdot-T)|_{[0,\infty)}$ is a 
solution of DI \eqref{di2}. Therefore $x^*(\cdot)|_{[-T,0]}$ is absolutely continuous and for almost every $t\in[-T,0]$, 
\begin{equation*}
 \frac{dx^*(t)}{dt}\in\hat{H}(x^*(t)).
\end{equation*}
Since $T>0$, is arbitrary, we have that $x^*(\cdot)|_{(-\infty,0]}$ is a solution on $(-\infty,0]$ of DI \eqref{di2}.

Therefore every limit point of $\left\{\Theta^{t}(\bar{x})\right\}_{t\geq0}$ is in $\Sigma$, the set of solutions of DI \eqref{di2}. Then by Theorem 
4.1 in \cite{benaim1} we get the following result.

\begin{theorem}\emph{[APT]}\label{apt}
 Under assumptions $(A1)-(A5)$, for almost every $\omega$, the linearly interpolated trajectory of recursion \eqref{rec}, $\bar{x}(\cdot)$, 
 is an asymptotic pseudotrajectory of DI \eqref{di2}.
\end{theorem}

\subsection{Characterization of limit sets}
For every $\omega$, the limit set of recursion \eqref{rec}, denoted by $L(\bar{x})$ is defined as,
\begin{equation}
 L(\bar{x}):=\cap_{t\geq0}\overline{\left\{\bar{x}(q+t):q\geq0\right\}}.
\end{equation}

As a consequence of Theorem \ref{apt} we will be able to characterize the limit set of recursion \eqref{rec} in terms of the dynamics 
induced by $\hat{H}$. The notions of invariance, internal chain transitivity, attracting sets, basin of attraction and attractors are 
taken from \cite{benaim1}. We shall state here definitions of a few of the notions mentioned above for the sake of completeness.

The \bf{flow }\rm of DI$\eqref{di2}$ is given by the set valued map $\Phi:\mathbb{R}^d\times \mathbb{R}\rightarrow
\left\{\text{subsets of }\mathbb{R}^d\right\}$, where for every $(x,t)\in\mathbb{R}^d\times\mathbb{R}$, 
$\Phi(x,t):=\left\{\bf{x}\rm(t):\bf{x}\rm\in\Sigma(\it{x}\rm)\right\}$. For any $C\subseteq\mathbb{R}^d$, let $\omega_{\Phi}(C):=\cap_{t\geq0}
\overline{\Phi_{[t,\infty)}(C)}$, where $\Phi_{[t,\infty)}(C):=\cup_{(x,\tau)\in\ C\times[t,\infty)}\Phi(x,\tau)$. For any $A\subseteq\mathbb{R}^d$, 
its \bf{basin of attraction }\rm denoted by $B(A)$ is defined as, $B(A):=\left\{x\in\mathbb{R}^d:
\omega_{\Phi}\left(\left\{x\right\}\right)\subseteq A\right\}$.

A set $A\subseteq\mathbb{R}^d$ is said to be \bf{invariant }\rm for DI$\eqref{di2}$ if for all $x\in A$, there exists $\bf{x}\rm\in
\Sigma(\it{x}\rm)$ such that for every $t\in\mathbb{R}$, $\bf{x}\rm(t)\in A$.

A compact set $A\subseteq\mathbb{R}^d$ is an \bf{attracting set }\rm for the flow of DI$\eqref{di2}$ if there exists a neighborhood $O$ of $A$, 
with the property that for every $\epsilon>0$ there exists $t_{\epsilon}>0$ such that for every $t\geq t_{\epsilon}$, $\Phi_{t}(O)\subseteq 
N^{\epsilon}(A)$, where $N^{\epsilon}(A)$ stands for the $\epsilon$-neighborhood of $A$. 

A compact set $A\subseteq\mathbb{R}^d$ is an \bf{attractor }\rm for the flow of DI$\eqref{di2}$ if it is an attracting set and is invariant. 
Further if the basin of attraction of attractor $A$ is the whole of $\mathbb{R}^d$, that is $B(A)=\mathbb{R}^d$, then $A$ is a 
\bf{global attractor}\rm.

Given a set $A\subseteq\mathbb{R}^d$ and $x,y\in A$, for any $\epsilon>0$ and $T>0$ there exists an $(\epsilon,T)$ chain from $x$ to $y$ for 
DI\eqref{di2} if there exists an integer $n\in\mathbb{N}$, solutions $\bf{x}\rm_1,\dots,\bf{x}\rm_n$ to DI\eqref{di2} and real numbers $t_1,\dots,t_n$ 
greater than $T$ such that
\begin{itemize}
 \item for all $i\in\left\{1,\dots,n\right\}$ and for all $q\in[0,t_i]$, $\bf{x}\rm_i(\it{q}\rm)\in A$,
 \item for all $i\in\left\{1,\dots,n\right\}$, $\parallel\bf{x}\rm_i(\it{t_i}\rm)-\bf{x}\rm_{i+1}(0)\parallel\leq\epsilon$,
 \item $\parallel \bf{x}\rm_1(0)-x\parallel\leq\epsilon$ and $\parallel\bf{x}\rm_n(\it{t_n}\rm)-y\parallel\leq\epsilon$.
\end{itemize}
A compact set $A\subseteq\mathbb{R}^d$ is said to be \bf{internally chain transitive }\rm if for every $x,y\in A$, for every $\epsilon>0$ and 
for every $T>0$, there exists $(\epsilon,T)$ chain from $x$ to $y$ for the DI\eqref{di2}.

\begin{theorem}\emph{[Limit set]}\label{ls}
 Under assumptions $(A1)-(A5)$, for almost every $\omega$, the following hold.
 \begin{itemize}
  \item [(i)] $L(\bar{x})$ is a non-empty, compact subset of $\mathbb{R}^d$ and is internally chain transitive. 
  \item [(ii)] If $A\subseteq\mathbb{R}^d$ is an attracting set for the flow of DI \eqref{di2} with a basin of attraction $B(A)$, such 
  that $L(\bar{x})\cap B(A)\neq\emptyset$, then $L(\bar{x})\subseteq A$.
  \item [(iii)] If $A\subseteq\mathbb{R}^d$ is a global attractor for the flow of DI \eqref{di2}, then $L(\bar{x})\subseteq A$.
  \item [(iv)] If $A=\left\{x^*\right\}$ is a global attractor for the flow of DI \eqref{di2}, then $\left\{x_n\right\}_{n\geq0}$ converges to $x^*$.
 \end{itemize}
\end{theorem}
\begin{proof}:
 \begin{itemize}
  \item [(i)] Follows from Theorem \ref{apt} above and Theorem 4.3 of \cite{benaim1}.
  \item [(ii)] From part $(i)$ of this theorem we have that $L(\bar{x})$ is internally chain transitive. Now the claim follows from Theorem 
  3.23 of \cite{benaim1}.
  \item [(iii)] Follows from part $(i)$ of this theorem and Corollary 3.24 of \cite{benaim1}.
  \item [(iv)] Follows from part $(iii)$ of this theorem with $A={x^*}$.\qed
 \end{itemize}
\end{proof}

\section{Iterate-independent and controlled Markov noise cases}
\label{siacmnc}
In this section we shall introduce two variants of assumption $(A2)$ under which a similar analysis as presented in the previous section 
can be carried out (assumptions $(A1),\ (A3)-(A5)$ are assumed to hold). The first case is of iterate-independent Markov noise i.e., 
where the transition probabilities do not depend on the iterate $x$ while second case is of controlled Markov noise i.e., the 
so-called noise process is not Markov by itself, but its lack of Markov property comes through its dependence on some possibly imperfectly known 
time varying process which is viewed as a control for analysis purposes. Throughout this section we assume that 
$(A1),\ (A3)-(A5)$ hold.

\subsection{Iterate-independent Markov noise}
First present the iterate-independent Markov noise assumption.
\begin{itemize}
 \item [$(A2)'$] $\left\{S_n\right\}_{n\geq0}$ is a sequence of $\mathcal{S}$ valued random variables on $\Omega$ such that, 
  for every $n\geq0$, for every $A\in \mathscr{B}(\mathcal{S})$, $\mathbb{P}(S_{n+1}\in A|S_m,X_m,m\leq n)=
  \mathbb{P}(S_{n+1}\in A|S_n)=\Pi(S_n)(A)$ $a.s.$ where $\Pi:\mathcal{S}\rightarrow\mathcal{P}(\mathcal{S})$ is continuous.
\end{itemize}      

Assumption $(A2)'$ and the fact that the state space of Markov noise is a compact metric space gives us that 
the Markov chain defined by the transition kernel, $\Pi(\cdot)(\cdot)$, admits a stationary distribution ($\mu\in\mathcal{P}(\mathcal{S})$ is 
stationary for the Markov chain defined by the transition kernel $\Pi(\cdot)(\cdot)$ if, for every $A\in\mathscr{B}(\mathcal{S})$, 
$\mu(A)=\int_{\mathcal{S}}\Pi(s)(A)\mu(ds)$).   
Let $D\subseteq\mathcal{P}(\mathcal{S})$ denote the set of stationary distributions of the Markov chain defined by the transition kernel 
$\Pi(\cdot)(\cdot)$. It is easy to show that $D$ is a convex and compact subset of $\mathcal{P}(\mathcal{S})$. Define the set valued map, 
$\hat{H}_{i}:\mathbb{R}^d\rightarrow\left\{subsets\ of\ \mathbb{R}^d\right\}$, such that, for every $x\in\mathbb{R}^d$, 
\begin{equation*}
 \hat{H}_{i}(x):=\cup_{\mu\in D}\int_{\mathcal{S}}H_x(s)\mu(ds),
\end{equation*}
where for every $x\in\mathbb{R}^d$, the set valued map $H_x$ is as defined in section \ref{ldi}. Similar to Lemma \ref{march2}, it can be 
shown that the set valued map $\hat{H}$ defined above is a Marchaud map. By the same set of arguments as in section \ref{recanal}, we obtain 
the following result. 

\begin{theorem}\emph{[APT for iterate-independent Markov noise]}
 Under assumptions $(A1),\ (A2)',\ (A3)-(A5)$, for almost every $\omega$, the linearly interpolated trajectory of recursion \eqref{rec}, is 
 an asymptotic pseudotrajectory for the flow of the DI,
 \begin{equation}\label{di3}
  \frac{dx}{dt}\in\hat{H}_{i}(x).
 \end{equation}
\end{theorem}

As a consequence of the above theorem, a limit set theorem similar to Theorem \ref{ls} can be established which under assumptions $(A1),(A2)',
(A3)-(A5)$ characterizes the limit set of recursion \eqref{rec} in terms of the dynamics of DI\eqref{di3}.

\subsection{Controlled Markov noise}
Let $\left\{Z_n\right\}_{n\geq0}$ be a sequence of random variables(control sequence) on $\Omega$, taking values in a compact metric space, 
$\mathcal{Z}$. The controlled Markov noise assumption is as follows.
\begin{itemize}
 \item [$(A2)''$]$\left\{S_n\right\}_{n\geq0}$ is a sequence of $\mathcal{S}$ valued random variables on $\Omega$ such that, 
  for every $n\geq0$, for every $A\in \mathscr{B}(\mathcal{S})$, $\mathbb{P}(S_{n+1}\in A|S_m,Z_m,X_m,m\leq n)=
  \mathbb{P}(S_{n+1}\in A|S_n,Z_n,X_n)=\Pi(X_n,S_n,Z_n)(A)$ a.s. with $\Pi:\mathcal{S}\rightarrow\mathcal{P}(\mathcal{S})$, 
  satisfying,
      \begin{itemize}
      \item [(i)] $\Pi$ is continuous.
      \item [(ii)] $\left\{Z_n\right\}_{n\geq0}$ is a stationary randomized control, i.e., for every $n\geq0$, for every $A\subseteq\mathscr{B}(
      \mathcal{Z}$), $\mathbb{P}(Z_n\in A|X_m,S_m,Z_{m-1},m\leq n)=\mathbb{P}(Z_n\in A|S_n)=\phi(S_n)(A)$ $a.s.$ with $\phi:\mathcal{S}
      \rightarrow\mathcal{P}(Z)$, measurable.
      \end{itemize}
  Let $\Pi_{\phi}:\mathbb{R}^d\times\mathcal{S}\rightarrow\mathcal{P}(\mathcal{S})$ be such that for every $x\in\mathbb{R}^d$, for every 
  $s\in\mathcal{S}$, $\Pi_{\phi}(x,s)(ds')=\int_{\mathcal{Z}}\Pi(x,z,s)(ds')\phi(s)(dz)$.    
      \begin{itemize}
      \item [(iii)] For every $x\in\mathbb{R}^d$ the Markov chain with state space $\mathcal{S}$ defined by the transition kernel, 
      $\Pi_{\phi}(x,\cdot)(\cdot)$  admits 
      a stationary distribution ($\mu\in\mathcal{P}(\mathcal{S})$ is stationary for the Markov chain defined by the transition kernel 
      $\Pi_{\phi}(x,\cdot)(\cdot)$ if, for every $A\in\mathscr{B}(\mathcal{S})$, $\mu(A)=\int_{\mathcal{S}}\Pi_{\phi}(x,s)(A)\mu(ds)$).   
      \end{itemize}
\end{itemize}
For every $x\in\mathbb{R}^d$, let $D_{\phi}(x)\subseteq\mathcal{P}(\mathcal{S})$ denote the set of stationary distributions of the Markov 
chain defined by the transition kernel, $\Pi_{\phi}(x,\cdot)(\cdot)$. The convexity, compactness of $D_{\phi}(x)$ and the closed graph 
property of the set valued map $x\rightarrow D_{\phi}(x)$ follow from arguments similar to those in page 69 of \cite{borkartxt}. Define the 
set valued map, $\hat{H}_{\phi}:\mathbb{R}^d\rightarrow\left\{\text{subsets of }\mathbb{R}^d\right\}$ such that for every $x\in\mathbb{R}^d$,
\begin{equation*}
 \hat{H}_{\phi}(x):=\cup_{\mu\in D_{\phi}(x)}\int_{\mathcal{S}}H_x(s)\mu(ds),
\end{equation*}
where for every $x\in\mathbb{R}^d$, the set valued map $H_x$ is as defined in section \ref{ldi}.
By a proof similar to Lemma \ref{march2}, one can show that the map $\hat{H}_{\phi}$ is a Marchaud map. Then by a similar set of arguments 
as in section \ref{recanal}, one obtains the following result.

\begin{theorem}\emph{[APT for controlled Markov noise]}
 Under assumptions $(A1),\ (A2)'',\ (A3)-(A5)$, for almost every $\omega$, the linearly interpolated trajectory of recursion \eqref{rec} is 
 an asymptotic pseudotrajectory for the flow of the DI,
 \begin{equation}\label{di4}
  \frac{dx}{dt}\in\hat{H}_{\phi}(x).
 \end{equation}
\end{theorem}
As a consequence of the above theorem, a limit set theorem similar to Theorem \ref{ls} can be established which under assumptions $(A1),(A2)'',
(A3)-(A5)$ characterizes the limit set of recursion \eqref{rec} in terms of the dynamics of DI\eqref{di4}.

\section{Applications}
\label{app}
In this section we consider four applications of the foregoing. These applications presented are natural extensions of the ones presented 
in chapter 5.3 of \cite{borkartxt} to the case with Markov noise. All the applications presented in this section are stated with Markov noise 
terms $\left\{S_n\right\}_{n\geq0}$ satisfying assumption $(A2)$. One may also consider the applications presented here with the Markov noise 
terms satisfying assumptions $(A2)'$ or $(A2)''$ as well.

\subsection{Controlled stochastic approximation}
Consider the iteration 
\begin{equation}\label{reccsa}
 X_{n+1}-X_n-a(n)M_{n+1}=a(n)h(X_n,Z_n,S_n),
\end{equation}
where, $\{Z_n\}_{n\geq0}$ is a random sequence taking values in a compact metric space $\mathcal{Z}$ and $h:\mathbb{R}^d\times\mathcal{Z}
\times\mathcal{S}\rightarrow\mathbb{R}^d$ is continuous and Lipschitz in the first argument uniformly w.r.t. the second and third, i.e., 
there exists $L>0$ such that, for every $(z,s)\in\mathcal{Z}\times\mathcal{S}$ and for every $x_1,x_2\in\mathbb{R}^d$, 
\begin{equation*}
 \parallel h(x_1,z,s)-h(x_2,z,s)\parallel\leq L\parallel x_1-x_2\parallel.
\end{equation*}
$\{Z_n\}_{n\geq0}$ is viewed here as a control sequence, that is, $Z_n$ is chosen by the agent running the algorithm at time $n\geq0$ based on 
the observed history and possibly extraneous independent randomization as is usual in stochastic control problems. It could also be an 
unknown random process that affects the measurements in addition to the additive noise $\{M_n\}_{n\geq0}$. We shall assume that 
assumptions $(A2)-(A5)$ hold. 

Define the set valued map $H:\mathbb{R}^d\times\mathcal{S}\rightarrow\{\text{subsets of }\mathbb{R}^d\}$ such that, for every $(x,s)\in
\mathbb{R}^d\times\mathcal{S}$, 
\begin{equation*}
 H(x,s):=\bar{co}(\left\{h(x,z,s):z\in\mathcal{Z}\right\}).
\end{equation*}
Then it can be shown that the set valued map $H$ satisfies assumption $(A1)$ (see Lemma 5 in chapter 5.3 of \cite{borkartxt}). Then 
recursion \eqref{reccsa} can be written as,
\begin{equation*}
 X_{n+1}-X_n-a(n)M_{n+1}\in a(n)H(X_n,S_n)
\end{equation*}
and the asymptotic behavior of the same can be analyzed as in section \ref{recanal} of this paper.

\subsection{Stochastic subgradient descent of a parametrized convex function}
Consider the function, $J:\mathbb{R}^d\times\mathcal{S}\rightarrow\mathbb{R}$, continuous and for every $s\in\mathcal{S}$, $J(\cdot,s)$ 
is a convex function. Suppose $J(\cdot,s)$ is not continuously differentiable everywhere, so that $\nabla J(x,s)$ is not defined at all 
$x\in\mathbb{R}^d$, where $\nabla$ denotes gradient w.r.t. $x$. A natural generalization of gradient to this non-smooth case is the notion 
of subdifferential w.r.t. $x$ denoted by $\partial J(x,s)$, defined as the set of all $y\in\mathbb{R}^d$ such that for every 
$z\in\mathbb{R}^d$, 
\begin{equation*}
 J(z,s)\geq J(x,s)+\langle y,z-x\rangle.
\end{equation*}
Then for every $(x,s)\in\mathbb{R}^d\times\mathcal{S}$, $\partial J(x,s)$ is nonempty, convex and compact (see Proposition 5.4.1 in 
\cite{bertsekas}). From the definition of subdifferential and the continuity of function $J(\cdot)$, it can be easily shown that the set 
valued map $(x,s)\rightarrow\partial J(x,s)$ has a closed graph.

Define the set valued map $H:\mathbb{R}^d\times\mathcal{S}\rightarrow\{\text{subsets of }\mathbb{R}^d\}$, such that for every 
$(x,s)\in\mathbb{R}^d\times\mathcal{S}$, 
\begin{equation}\label{sbgrd}
H(x,s)=-\partial J(x,s). 
\end{equation}
We further assume that the set valued map $H$ satisfies the linear growth property (that is assumption $(A1)(ii)$). Then by arguments in 
the preceding paragraph the set valued map $H$ satisfies assumption $(A1)$. 

Consider the recursion \eqref{rec} with the set valued map defined above. The recursion, 
\begin{equation*}
 X_{n+1}-X_{n}-a(n)M_{n+1}\in a(n)H(X_n,S_n),
\end{equation*}
where $\{M_{n}\}_{n\geq1}$ can be interpreted as subgradient estimation error. One case of interest where assumption $(A4)$ associated 
with the additive noise terms is satisfied is when $M_n$ are i.i.d. zero mean with finite variance (more generally martingale difference 
terms satisfying assumption $(A3)$ in chapter 2.1 of \cite{borkartxt}).

Suppose the assumptions $(A2)-(A5)$ are satisfied then the asymptotic behavior of the above recursion can be analyzed as 
in section \ref{recanal} of this paper. If in addition to assumption $(A2)$, we know that for every $x\in\mathbb{R}^d$, 
$|D(x)|=1$ (that is the stationary distribution is unique for every $x$), then the set valued map associated with the limiting DI, 
$\hat{H}$ as defined in \eqref{lim1} is such that for every $x\in\mathbb{R}^d$, 
\begin{equation*}
\hat{H}(x)\subseteq-\partial\mathbb{E}_{s}[J(\cdot);\mu_x](x), 
\end{equation*}
where $\{\mu_x\}=D(x)$, $\mathbb{E}_{s}[J(\cdot);\mu_x]$ denotes a real valued map on $\mathbb{R}^d$ such that, for every 
$y\in\mathbb{R}^d$, $\mathbb{E}_{s}[J(\cdot);\mu_x](y)=\int_{\mathcal{S}}J(y,s)\mu_x(ds)$ (which is clearly a convex function) and 
$\partial\mathbb{E}_{s}[J(\cdot);\mu_x](x)$ denotes the subdifferential of the above defined map at $x$. If in addition we have that the 
state space of the Markov noise is finite (that is $|\mathcal{S}|<\infty$) then by Proposition 5.4.6 in \cite{bertsekas} we have that for 
every $x\in\mathbb{R}^d$, 
\begin{equation*}
\hat{H}(x)=-\partial\mathbb{E}_{s}[J(\cdot);\mu_x](x).
\end{equation*}
As a consequence of the above we have that the global attractor of the DI,
\begin{equation*}
 \frac{dx}{dt}\in-\partial\mathbb{E}_{s}[J(\cdot);\mu_x](x),
\end{equation*}
is also a global attractor of the limiting DI (that is DI \eqref{di2})and hence by the limit set theorem (that is Theorem \ref{ls}), we 
can conclude that the iterates $\{X_n\}_{n\geq0}$ will almost surely converge to such an attractor.

Another important application is the minimization of the convex function obtained by averaging the parameter w.r.t. a particular 
probability distribution say, $\mu\in\mathcal{P}(\mathcal{S})$. Formally the minimization of function $J_{\mu}:\mathbb{R}^d\rightarrow
\mathbb{R}$, where for every $x\in\mathbb{R}^d$, $J_{\mu}(x):=\int_{\mathcal{S}}J(x,s)\mu(ds)$. Consider recursion \eqref{rec}, where the 
set-valued map $H$ is as in equation \eqref{sbgrd} and Markov noise terms (that is $\{S_n\}_{n\geq0}$) satisfy assumption $(A2)$ 
with the additional condition that for every $x\in\mathbb{R}^d$, $D(x)=\{\mu\}/D_{\phi}(x)=\{\mu\}$. Then under additional assumptions of 
$(A1),(A3)-(A5)$ we have that the linearly interpolated trajectory of recursion \eqref{rec} is an APT for the flow of DI \eqref{di2}. The set 
valued map $\hat{H}$ associated with DI \eqref{di2}, can be shown to satisfy,
\begin{equation*}
 \hat{H}(x)\subseteq-\partial J_{\mu}(x),
\end{equation*}
for every $x\in\mathbb{R}^d$. Further if the state space of the Markov noise is finite (that is $|\mathcal{S}|<\infty$) then by 
Proposition 5.4.6 in \cite{bertsekas} we have that for every $x\in\mathbb{R}^d$,
\begin{equation*}
 \hat{H}(x)=-\partial J_{\mu}(x).
\end{equation*}

\subsection{Approximate drift problem}
Consider the continuous function $h:\mathbb{R}^d\times\mathcal{S}\rightarrow\mathbb{R}^d$, satisfying the linear growth property, that is, 
there exists $K>0$, such that for every $(x,s)\in\mathbb{R}^d\times\mathcal{S}$, $\parallel h(x,s)\parallel\leq 
K(1+\parallel x\parallel)$. In many applications one wishes to implement the recursion,
\begin{equation*}
 X_{n+1}-X_{n}-a(n)M_{n+1}=a(n)h(X_n,S_n),
\end{equation*}
where assumptions $(A2)/(A2)'/(A2)'',(A3)-(A5)$ are assumed to hold. But in practice one has access only to approximate value of the 
function $h(\cdot)$. Further the magnitude of the error in approximation of $h(\cdot)$ is only known to be bounded by a constant, say 
$\epsilon>0$. Then the recursion implemented becomes,
\begin{equation}\label{recapd}
 X_{n+1}-X_{n}-a(n)M_{n+1}=a(n)(h(X_n,S_n)+\eta_n),
\end{equation}
where $\eta_n$ represents the error in estimation of $h(\cdot)$ and we have that for every $n\geq0$, $\parallel\eta_n\parallel\leq
\epsilon$.

Define the set valued map $H:\mathbb{R}^d\times\mathcal{S}\rightarrow\{\text{subsets of }\mathbb{R}^d\}$ such that for every $(x,s)\in
\mathbb{R}^d\times\mathcal{S}$, 
\begin{equation}\label{hapd}
 H(x,s):=h(x,s)+\epsilon U,
\end{equation}
where $U$ denotes the closed unit ball in $\mathbb{R}^d$. Then recursion \eqref{recapd} can be written as,
\begin{equation*}
 X_{n+1}-X_{n}-a(n)M_{n+1}\in a(n)H(X_n,S_n),
\end{equation*}
where the set valued map $H$ is as defined in equation \eqref{hapd}. It can be easily shown that the set valued map $H$ satisfies 
assumption $(A1)$. Under assumptions $(A2)-(A5)$ one may analyze the asymptotic behavior of recursion \eqref{recapd} as 
in section \ref{recanal} of this paper.

\subsection{Discontinuous dynamics}
Consider the recursion,
\begin{equation}\label{recdd}
 X_{n+1}-X_{n}-a(n)M_{n+1}=a(n)h(X_n,S_n), 
\end{equation}
where $h(\cdot)$ is merely measurable, satisfying the linear growth property that is there exists $K>0$, such that for every 
$(x,s)\in\mathbb{R}^d\times\mathcal{S}$, $\parallel h(x,s)\parallel\leq K(1+\parallel x\parallel)$. 

Define the set valued map $H:\mathbb{R}^d\times\mathcal{S}\rightarrow\{\text{subsets of }\mathbb{R}^d\}$ such that for every $(x,s)\in
\mathbb{R}^d\times\mathcal{S}$,
\begin{equation}\label{hdd}
 H(x,s):=\cap_{\epsilon>0}\bar{co}\left(\left\{h(y,s'):\parallel x-y\parallel<\epsilon,\ d_{\mathcal{S}}(s,s')<\epsilon\right\}\right).
\end{equation}
Then the set valued map defined above can be shown to satisfy assumption $(A1)$. Now the recursion \eqref{recdd} can be written in the form 
of recursion \eqref{rec} with the set valued map $H$ as defined in equation \eqref{hdd}. Under assumptions $(A2)-(A5)$, 
one may analyze the asymptotic behavior of recursion \eqref{recdd} as in section \ref{recanal} of this paper. 
\section{Conclusions and directions for future work}
\label{cadffw}
We have shown that almost surely the linearly interpolated trajectory of recursion \eqref{rec} is an asymptotic pseudotrajectory for the flow 
of DI \eqref{di2}. The asymptotic pseudotrajectory result proved in this paper enables one to characterize limit sets of recursion \eqref{rec} 
as internally chain transitive sets of the flow of DI \eqref{di2} which in-turn enables us to guarantee convergence of the iterates to 
attractors of DI \eqref{di2} as stated in Theorem \ref{ls}. We have also stated two variants of the Markov noise assumption where the analysis 
is very similar to the one presented. Finally applications are presented where the recursion studied in this paper 
naturally appears and in the case of subgradient descent we have also been able to interpret the set-valued map associated with DI \eqref{di2} as 
subdifferential of a certain averaged function. 

Certain extensions and applications which we wish to consider in future are listed below.
\begin{itemize}
 \item [(1)] Stochastic approximation schemes on multiple time scales are extensively studied in literature and find use in several 
 reinforcement learning and optimization applications. For the case without Markov noise, the two time scale stochastic approximation with 
 single valued maps and the case of set-valued maps have already been studied (see chapter 6 in \cite{borkartxt} for the single valued case 
 and \cite{leslep} for the set-valued case). One can extend the analysis presented in this paper to obtain a similar result for two time scale 
 stochastic recursive inclusions with Markov noise.
 \item [(2)] Throughout this paper we have assumed that the state space of the Markov noise is a compact metric space. In many practical 
 scenarios, the Markov noise terms take values in a general Polish space. For such a case we believe that the analysis presented in the paper 
 can be extended under certain additional assumptions. The idea would be to embed the polish space as a dense subset of a compact metric space 
 and under some additional assumptions which guarantee an upper semicontinuous extension of set-valued map $H$ and tightness of some 
 probability measures one might be able to provide an APT argument as in this paper (see \cite{borkarmark}).
 \item [(3)] Assumption $(A5)$, which ensures that the iterates are stable usually is hard to verify. The Borkar-Meyn theorem (see 
 \cite{borkarmeyn}) which provides a sufficient condition for stability of iterates for stochastic approximation schemes with single valued 
 maps has been extended to the set-valued case without Markov noise (see \cite{arunstab}) and we believe that a similar extension is possible 
 to the set-valued case with Markov noise. Stability results similar to the one in \cite{andvih} also need to be investigated for possible 
 extension to the set-valued case.
 \item [(4)] In many applications one encounters the recursion,
 \begin{equation*}
  X_{n+1}\in \text{\bf{P}\rm}\left[X_{n}+a(n)(H(X_n,S_n)+M_{n+1})\right],
 \end{equation*}
where \bf{P }\rm denotes the projection of the set $X_{n}+a(n)(H(X_n,S_n)+M_{n+1})$ onto a compact (possibly convex) subset of $\mathbb{R}^d$. Such 
recursions are of interest also because the assumption of stability (that is $(A5)$) is naturally satisfied. For the single valued case such 
recursions are studied in \cite{dupuis} and the extension of the same to the recursion above needs to be established.
\item [(5)] Another variant of Markov noise assumptions studied for the single valued case is given in \cite{matti}. These set of assumptions 
enable one to study behavior of stochastic approximation schemes in the single valued case when the state space is a general Euclidean space. 
Since in this paper we study the set-valued case by converting it to the single valued setting it would be interesting to see if those 
assumptions can be used for the set-valued case as in this paper.
\end{itemize}

\appendix
\section{Proof of the continuous embedding lemma (Lemma \ref{ctem})}
\label{prf_ctem}
The proof of this lemma is similar to the proof of Theorem 1, chapter 1, section 13 of \cite{aubindi}. We shall provide a brief outline here 
for the sake of completeness.

For any $\epsilon>0$, for every $(x_0,s_0)\in\mathbb{R}^d\times\mathcal{S}$, let 
$B(\epsilon,x_0,s_0):=\left\{(x,s):\parallel x-x_0\parallel<\epsilon,\ d_{\mathcal{S}}(s,s_0)<\epsilon\right\}$. Let 
$\left\{\epsilon_l:=\frac{1}{3^l}\right\}_{l\geq1}$. Then for every $l\geq1$, $\mathscr{C}_l:=\left\{B(\epsilon_l,x_0,s_0):(x_0,s_0)\in
\mathbb{R}^d\times\mathcal{S}\right\}$ is an open covering of $\mathbb{R}^d\times\mathcal{S}$. Since $\mathbb{R}^d\times\mathcal{S}$ is a 
metric space, it is paracompact (see Defn. 4 and Theorem 1 in chapter 0, section 1 of \cite{aubindi}). Therefore for every $l\geq1$, there 
exists a locally finite open refinement of the covering $\mathscr{C}_l$ and let it be denoted by $\tilde{\mathscr{C}}_l:=
\left\{C_i^l\right\}_{i\in I^{l}}$ where $I^l$ is an arbitrary index set. By Theorem 2, chapter 0, section 1 of \cite{aubindi}, there exists a 
continuous partition of unity, $\left\{\psi_i^{l}\right\}_{i\in I^l}$, subordinated to the covering $\tilde{\mathscr{C}}_l$. Therefore, 
for every $l\geq1$, for every $i\in I^l$, there exists $(x_i^l,s_i^l)$, such that $support(\psi_i^l)\subseteq C_i^l\subseteq 
B(\epsilon_l,x_i^l,s_i^l)$. For every $l\geq1$, for every $(x,s)$, let $I^l(x,s):=\left\{i\in I^l:\psi_i^l(x,s)>0\right\}$ and by definition 
of $\psi_i^l$, we have that $|I^l(x,s)|<\infty$ and $\sum_{i\in I^l(x,s)}\psi_i^l(x,s)=1$.

For every $l\geq1$, define the set valued map $H^{(l)}:\mathbb{R}^d\times\mathcal{S}\rightarrow\left\{\text{subsets of }\mathbb{R}^d\right\}$, 
such that for every $(x,s)$, $H^{(l)}(x,s):=\sum_{i\in I^l(x,s)}\psi_i^{l}(x,s)A_i^{l}$, where $A_i^{l}:=\bar{co}\left(H\left(B\left(2\epsilon_l,
x_i^l,s_i^l\right)\right)\right)$.

Fix $l\geq1$. For every $i\in I^l$, by assumption $(A1)(ii)$, $\sup_{z\in H\left(B\left(2\epsilon_l,x_i^l,s_i^l\right)\right)}\parallel z
\parallel\leq K(1+\parallel x_i^l\parallel +2\epsilon_l)$ from which we can deduce that $\sup_{z\in A_i^l}\parallel z\parallel \leq 
K(1+\parallel x_i^l\parallel+2\epsilon_l)$. Hence $A_i^l$ is convex and compact. Therefore $H^{(l)}(x,s)$ is a convex combination of compact 
and convex subsets of $\mathbb{R}^d$ and hence is \bf{convex and compact}\rm.

Fix $l\geq1$ and $(x,s)$. Then for every $i\in I^l(x,s)$, $(x,s)\in support(\psi_i^l)\subseteq C_i^l\subseteq B(\epsilon_l,x_i^l,s_i^l)\subseteq
B(2\epsilon_l,x_i^l,s_i^l)$. Therefore for every $i\in I^l(x,s)$, $H(x,s)\subseteq H\left(B\left(2\epsilon_l,x_i^l,s_i^l\right)\right)\subseteq 
A_i^l$ and since $H(x,s)$ is convex we have that $H(x,s)=\sum_{i\in I^l(x,s)}\psi_i^l(x,s)H(x,s)\subseteq\sum_{i\in I^l(x,s)}\psi_i^l(x,s)A_i^l=
H^{(l)}(x,s)$. Thus $\bf{H(x,s)\subseteq H^{(l)}(x,s)}\rm$.

Fix $l\geq1$ and $(x,s)$. Then for every $i\in I^{l+1}(x,s)$, $\parallel x-x_i^{l+1}\parallel<\epsilon_{l+1}$ and $d_{\mathcal{S}}(s,s_i^{l+1})
<\epsilon_{l+1}$. Similarly for every $j\in I^{l}(x,s)$, $\parallel x-x_j^{l}\parallel<\epsilon_{l}$ and $d_{\mathcal{S}}(s,s_j^{l})
<\epsilon_{l}$. Therefore for every $i\in I^{l+1}(x,s)$, for every $j\in I^l(x,s)$, $\parallel x_i^l-x_i^{l+1}\parallel<\epsilon_l+\epsilon_{l+1}=
\epsilon_l+\frac{\epsilon_l}{3}=\frac{4\epsilon_l}{3}$ and $d_{\mathcal{S}}(s_i^l,s_i^{l+1})<\frac{4\epsilon_l}{3}$. For every $i\in 
I^{l+1}(x,s)$, for every $j\in I^{l}(x,s)$, for every $(x',s')\in B(2\epsilon_{l+1},x_i^{l+1},s_i^{l+1})$, $\parallel x'-x_j^{l}\parallel\leq
\parallel x'-x_i^{l+1}\parallel+\parallel x_i^{l+1}-x_j^l\parallel<2\epsilon_{l+1}+\frac{4\epsilon_l}{3}=\frac{2\epsilon_l}{3}+
\frac{4\epsilon_l}{3}=2\epsilon_l$ and $d_{\mathcal{S}}(s',s_j^{l})<2\epsilon_l$. Thus for every $i\in I^{l+1}(x,s)$, for every 
$j\in I^{l}(x,s)$, for every $(x',s')\in B(2\epsilon_{l+1},x_i^{l+1},s_i^{l+1})$, $H(x',s')\subseteq H\left(B\left(2\epsilon_l,
x_j^l,s_j^l\right)\right)\subseteq A_j^l$. Therefore for every $i\in I^{l+1}(x,s)$, for every $j\in I^{l}(x,s)$, $A_i^{l+1}\subseteq A_j^l$ 
and by using convexity of $A_i^{l+1}$ we get that for every $i\in I^{l+1}(x,s)$, $A_i^{l+1}=\sum_{j\in I^{l}(x,s)}\psi_j^{l}(x,s)A_i^{l+1}\subseteq
\sum_{j\in I^{l}(x,s)}\psi_j^{l}(x,s)A_j^{l}=H^{(l)}(x,s)$. By using convexity of $H^{(l)}(x,s)$ we get that 
$H^{(l+1)}(x,s)=\sum_{i\in I^{l+1}(x,s)}\psi_i^{l+1}(x,s)A_i^{l+1}\subseteq\sum_{i\in I^{l+1}(x,s)}\psi_i^{l+1}(x,s)H^{(l)}(x,s)=H^{(l)}(x,s)$. 
Therefore, $\bf{H^{(l+1)}(x,s)\subseteq H^{(l)}(x,s)}\rm$.

Fix $(x,s)$. Clearly $H(x,s)\subseteq \cap_{l\geq1}H^{(l)}(x,s)$. By u.s.c. of $H$, we have that for every $\epsilon>0$, there exists 
$\delta(\epsilon,x,s)>0$ such that for every $(x',s')$ satisfying $\parallel x'-x\parallel<\delta(\epsilon,x,s)$ and $d_{\mathcal{S}}(s',s)<
\delta(\epsilon,x,s)$ we have $H(x',s')\subseteq H(x,s)+\epsilon U$. Clearly we can find $L$ such that for every $l\geq L$, $3\epsilon_l<
\delta(\epsilon,x,s)$. Therefore for every $\epsilon>0$, for $L$ as above, we have that for every $l\geq L$, for every $i\in I^l(x,s)$, 
$H\left(B\left(2\epsilon_l,x_i^l,s_i^l\right)\right)\subseteq H(x,s)+\epsilon U$. Since $H(x,s)+\epsilon U$ is convex we get, for every 
$\epsilon>0$, there exists $L$ such that for every $l\geq L$, for every $i\in I^l(x,s)$, $A_i^l\subseteq H(x,s)+\epsilon U$. Again by convexity 
of $H(x,s)+\epsilon U$, we have that for every $\epsilon>0$, there exists $L$ such that $l\geq L$, $H^{(l)}(x,s)=\sum_{i\in I^l(x,s)}
\psi^l_i(x,s)A_i^l\subseteq H(x,s)+\epsilon U$. Therefore for every $\epsilon>0$, $\cap_{l\geq1}H^{(l)}(x,s)\subseteq H(x,s)+\epsilon U$. Now
it follows that $\bf{H(x,s)=\cap_{l\geq1}H^{(l)}(x,s)}\rm$.

Fix $l\geq1$ and $(x,s)$. We have already shown that for every $i\in I^l(x,s)$, $\sup_{z\in A_i^l}\parallel z\parallel\leq K(1+\parallel x_i^l
\parallel+2\epsilon_l)$. Using the fact that for every $i\in I^l(x,s)$, $\parallel x_i^l-x\parallel<\epsilon_l$ we get that for every $i\in 
I^l(x,s)$, $\sup_{z\in A_i^l}\parallel z\parallel\leq K(1+\parallel x\parallel+3\epsilon_l)$. Since $\left\{\psi_i^l(x,s)\right\}_{i\in 
I^l(x,s)}$ are convex combining coefficients, we get that $\bf{\sup_{z\in H^{(l)}(x,s)}\parallel z\parallel\leq K^{(l)}(1+\parallel x
\parallel)}\rm$, where $K^{(l)}:=K+3\epsilon_l$.

Fix $l\geq1$ and $(x,s)$. Since $\tilde{\mathscr{C}}_l$ is a locally finite open covering of $\mathbb{R}^d\times\mathcal{S}$, there exists 
$\delta>0$ such that $|I^l(\delta,x,s)|<\infty$ where $I^l(\delta,x,s):=\left\{i\in I^l:C_i^l\cap B(\delta,x,s)\neq\emptyset\right\}$. Since 
for every $i\in I^l$, $support(\psi_i^l)\subseteq C_i^l$, for every $(x',s')\in B(\delta,x,s)$ we have that $H^{(l)}(x',s')=\sum_{i\in 
I^l(\delta,x,s)}\psi_i^l(x',s')A_i^l$. Let $z\in H^{(l)}(x,s)$ and $\left\{\left(x_n,s_n\right)\right\}_{n\geq1}$ be some sequence converging 
to $(x,s)$. Then for every $i\in I^l(\delta,x,s)$ there exists $a_i^l\in A_i^l$ such that $z=\sum_{i\in I^l(\delta,x,s)}\psi_i^l(x,s)a_i^l$ 
and there exists $N$, such that for every $n\geq N$, $(x_n,s_n)\in B(\delta,x,s)$. For every $n\leq N$, define $z_n$ to be any value in 
$H^{(l)}(x_n,s_n)$ and for every $n\geq N$, define $z_n=\sum_{i\in I^l(\delta,x,s)}\psi_i^l(x_n,s_n)a_i^l\in H^{(l)}(x_n,s_n)$ where $a_i^l$ 
are as defined above. Then clearly for every $n\geq1$, $z_n\in H^{(l)}(x_n,s_n)$ and $\left\{z_n\right\}_{n\geq1}$ converges to $z$. Thus 
$H^{(l)}$ is l.s.c. at $(x,s)$. Upper simecontinuity of the map $H^{(l)}$ at $(x,s)$ follows from Theorem 1, chapter 1, section 13 of 
\cite{aubindi}. Since $(x,s)$ is arbitrary the map $H^{(l)}$ is both u.s.c. and l.s.c at every $(x,s)$ and hence is \bf{continuous}\rm. 

\section{Measurability of set-valued map $G$ in Lemma \ref{chint}}
\label{prf_msb}
Recall that $G:\mathcal{S}\rightarrow\left\{\text{subsets of }U\right\}$ is such that for every $s\in\mathcal{S}$, $G(s):=\left\{u\in U:
h^{(l)}(x,s,u)=f(s)\right\}$. Let $C\subseteq U$ be a closed subset. Then $G^{-1}(C)=\left\{s\in\mathcal{S}:G(s)\cap C\neq\emptyset\right\}$.
Thus $s\in G^{-1}(C)$ if and only if there exists $u\in C$, such that $h^{(l)}(x,s,u)=f(s)$ or equivalently $s\in G^{-1}(C)$ if and only if 
$f(s)\in h^{(l)}(x,s,C)$, where $h^{(l)}(x,s,C):=\left\{h^{(l)}(x,s,u):u\in C\right\}$. Since $h_x^{(l)}:\mathcal{S}\times C\rightarrow
\mathbb{R}^d$ is continuous, by Proposition 1, chapter 1, section 2 of \cite{aubindi}, the set valued map $s\rightarrow h_x^{(l)}(s,C)
$ is continuous and its measurability now follows from arguments similar to Lemma \ref{msble}(ii). Therefore, $G^{-1}(C)=\left\{s\in\mathcal{S}
:f(s)\in h^{(l)}(x,s,C)\right\}=\left\{s\in\mathcal{S}:d(f(s),h^{(l)}(x,s,C))=0\right\}$. By Lemma \ref{mtool}, 
$\left\{s\in\mathcal{S}:d(f(s),h^{l}(x,s,C))=0\right\}\in\mathscr{B}(\mathcal{S})$ and hence $G^{-1}(C)\in\mathscr{B}(\mathcal{S})$.

\section{Proof of uniform convergence in Lemma \ref{march1}(iii)}
\label{prf_conv}
Since $\left\{x_{n_k}\right\}_{k\geq1}$ converges to $x$, there exists $r>0$, such that $\sup_{k\geq1}\parallel x_{n_k}\parallel\leq r$. 
Fix $\epsilon>0$. We know that $h^{(l)}(\cdot)|_{2rU\times \mathcal{S}\times U}$ is uniformly continuous and hence there exists $\delta>0$ 
(depending on $\epsilon$), such that for every $x'\in\mathbb{R}^d$ satisfying $\parallel x'-x\parallel<\delta$, 
$\sup_{(s,u)\in\mathcal{S}\times U}\parallel h_x'^{(l)}(s,u)-h_x^{(l)}(s,u)\parallel<\epsilon$. Hence there exists $K$, such that for every 
$k\geq K$, $\parallel x_{n_k}-x\parallel<\delta$ and therefore for every $k\geq K$, 
$\sup_{(s,u)\in\mathcal{S}\times U}\parallel h_{x_{n_k}}^{(l)}(s,u)-h_x^{(l)}(s,u)\parallel<\epsilon$.
\bibliographystyle{IEEEtran}
\bibliography{Ref}
\end{document}